\documentclass{article}
\usepackage{amsmath}
\usepackage{amsfonts}
\usepackage{amsthm}
\usepackage{mathrsfs}
\usepackage{cancel}
\usepackage{float}

\usepackage{amssymb}

\usepackage{xcolor}

\numberwithin{equation}{section}

 \def\lb{\label}
\def\be{\begin{equation}}
\def\ee{\end{equation}}

\newcommand{\oalpha}{\overline{\alpha}}
\newcommand{\obeta}{\overline{\beta}}
\newcommand{\oo}{\overline{0}}
\newcommand{\ol}{\overline{1}}

\newcommand{\mcU}{\mathcal{U}}
\newcommand{\mcA}{\mathcal{A}}
\newcommand{\mcB}{\mathcal{B}}
\newcommand{\hC}{\widehat{C}}
\newcommand{\hCad}{\widehat{C}_{\ad}}
\newcommand{\hCp}{\widehat{C}_+}
\newcommand{\hCm}{\widehat{C}_-}
\newcommand{\tC}{\widetilde{C}}
\newcommand{\tCm}{\widetilde{C}_-}
\newcommand{\bii}{\mathbf{1}}
\newcommand{\bI}{\mathbf{I}}
\newcommand{\bP}{\mathbf{P}}
\newcommand{\bPad}{\mathbf{P}^{(\operatorname{ad})}}
\newcommand{\bK}{\mathbf{K}}
\newcommand{\hI}{\widehat{I}}

\newcommand{\mcP}{\mathcal{P}}
\newcommand{\mcK}{\mathcal{K}}
\newcommand{\mfg}{\mathfrak{g}}

\newcommand{\oI}{\overline{I}}

\newcommand{\Vad}{V_{\ad}}
\newcommand{\kg}{{\sf g}}
\newcommand{\okg}{\overline{\sf g}}
\newcommand{\ovarepsilon}{\overline{\varepsilon}}

\DeclareMathOperator{\End}{End}

\DeclareMathOperator{\ad}{ad}
\DeclareMathOperator{\tr}{tr}
\DeclareMathOperator{\str}{str}

\DeclareMathOperator{\bstr}{\pmb{\str}}
\DeclareMathOperator{\Sym}{Sym}
\DeclareMathOperator{\proj}{P}

\DeclareMathOperator{\Mat}{Mat}

\DeclareMathOperator{\sdim}{sdim}

\DeclareMathAccent{\wtilde}{\mathord}{largesymbols}{"65}

\newtheorem{proposition}{Proposition}

\allowdisplaybreaks

\begin{document}

\begin{flushright}
\today\\
\end{flushright}
\vspace{2cm}

\begin{center}
{\Large \bf The split Casimir operator \\[8pt] and solutions of the Yang-Baxter equation \\[8pt] for the $osp(M|N)$ and $s\ell(M|N)$ Lie superalgebras, \\[8pt] higher Casimir operators, and the Vogel parameters}
\end{center}
\vspace{1cm}

\begin{center}
{\large \bf  A. P. Isaev${}^{a,b}$, A. A. Provorov${}^{a,c}$}
\end{center}

\vspace{0.2cm}

\begin{center}
{${}^a$ \it
Bogoliubov Laboratory of Theoretical Physics,
Joint Institute for Nuclear Research, Dubna, Moscow region,
Russia}\vspace{0.2cm}

${}^b$ {\it
Faculty of Physics,
M.~V.~Lomonosov Moscow State University,
Moscow, Russia}\vspace{0.2cm}

{${}^c$ \it Moscow Institute of Physics and Technology (National Research University), Dolgoprudny, Moscow Region, Russia}
\vspace{0.5cm}

{\tt isaevap@theor.jinr.ru, aleksanderprovorov@gmail.com}
\end{center}
\vspace{1cm}
\begin{abstract}\noindent
We find the characteristic identities for the split Casimir operator in the defining and adjoint representations of the $osp(M|N)$ and $s\ell(M|N)$ Lie superalgebras. These identities are used to build the projectors onto invariant subspaces of the representation $T^{\otimes 2}$ of the $osp(M|N)$ and $s\ell(M|N)$ Lie superalgebras  in the cases when $T$ is the defining and adjoint representations. For defining representations, the $osp(M|N)$- and $s\ell(M|N)$-invariant solutions of the Yang-Baxter equation are expressed as rational functions of the split Casimir operator. For the adjoint representation, the characteristic identities and invariant projectors obtained are considered from the viewpoint of a universal description of Lie superalgebras by means of the Vogel parametrization. We also construct a universal generating function for higher Casimir operators of the $osp(M|N)$ and $s\ell(M|N)$ Lie superalgebras in the adjoint representation.
\end{abstract}

\newpage

\section{Introduction}
It is known that the split Casimir operator $\hC$ (see definition in Sec. \ref{Sec2}; also see \cite{Okub}) plays an important role in the description of Lie algebras and superalgebras as well as in the study of their representations. Furthermore, the operator $\hC$ is used for constructing solutions of the semiclassical and quantum Yang-Baxter equations that are invariant under the action of Lie algebras and superalgebras in various representations (see, e.g., \cite{ChPr}, \cite{ZMa}).

In the present paper, we use the operator $\hC$ to construct a system of projectors onto invariant subspaces of the representations $T\otimes T$ of the complex Lie superalgebras $osp(M|N)$ and $s\ell(M|N)$ in the cases when $T=T_f$ is the defining representation and when $T=\ad$ is the adjoint representation.

The idea to construct projectors onto invariant supspaces of representations of Lie algebras and superalgebras by means of invariant operators is not new. For example, the invariant projectors that act on the tensor product of the $s\ell(N)$ Lie algebra defining representations are called the Young symmetrizers and are constructed as images of specific elements of the group algebra $\mathbb{C}[S_r]$ of the symmetric group $S_r$. The algebra $\mathbb{C}[S_r]$ centralises the action of the algebra $sl(N)$ in the representation $T^{\otimes r}_f$. For the $so(N)$ and $sp(N)$ Lie algebras (where $N=2n$ is even) there exists an analogous statement: the action of those algebras in the representation $T^{\otimes r}_f$ is centralized by the Brauer algebra $B_r(N)$ (see e.g. \cite{Book2}). The aforementioned properties of the $s\ell(N)$, $so(N)$, and $sp(N)$ Lie algebras are carried over to the case of Lie superalgebras: in \cite{SWsl} and \cite{SWsl2} the method of describing subrepresentations of $T_f^{\otimes r}$ by means of the Young symmetrizers was generalized to encompass the $s\ell(M|N)$ Lie superalgebras, and in \cite{SWosp} an analogous result was obtained for the $osp(M|N)$ Lie superalgebras. In our work, we consider a decomposition of the representation $T\otimes T$ into subrepresentations by using the operator $\hC$, that is defined uniformly for all Lie superalgebras with the non-degenerate Cartan-Killing metric. Within this approach the $s\ell(M|N)$- and $osp(M|N)$-Lie superalgebras are described in a similar fashion.

In the case where $T=\ad$ is the adjoint representation, the construction of projectors onto invariant subspaces of the representation $T\otimes T=\ad\otimes \ad$ by using $\hC$ has one more significance. It is related to the notion of the Universal Lie algebra, which was introduced by Vogel in \cite{Vog} (see also \cite{Del}, \cite{Lan2}). The Universal Lie algebra was supposed to be a model of all complex simple Lie algebras, embracing some Lie superalgebras additionally. For example, many quantities that characterize the Lie algebra $\mfg$ in different representations $T_\lambda$ (possibly reducible) that participate in the decomposition $\ad^{\otimes k}=\sum_{\lambda}T_\lambda$ where $k \ge 1$ are expressed as rational functions of the three Vogel parameters (see their definition in Sec. \ref{Sec5}). These parameters take specific values for all complex simple Lie algebras as well as for all basic classical Lie superalgebras (see, e.g., \cite{Mkr}, and
Sec.{\bf \ref{Sec5}} below). In particular, it was shown for Lie algebras that using the Vogel parameters one can express the dimensions of the representations $T_\lambda$, when $k=2,3$ \cite{Vog}, the dimensions of the ad-series representations, i.e., the representations $T_{\lambda'}$ with the highest weight $\lambda'=k\lambda_{\ad}$ where $\lambda_{\ad}$ is the highest weight of the given Lie algebra \cite{Lan}, the dimensions of the representations of $X_2$-series \cite{AvMkr} as well as the values of higher Casimir operators in the adjoint representation of the given Lie algebra \cite{Mkr}. Furthermore, in \cite{Knots}, \cite{Knots2} it was shown that the universal description of complex simple Lie algebras allows formulating some types of knot polynomials via characters simultaneously for all 
types of quantum simple Lie groups.

 The paper is organized as follows. In Section \ref{Sec2}, we recall the main notions of Lie superalgebra theory and introduce some necessary conventions that we use throughout the work. Sections \ref{Sec3} and \ref{Sec4} are dedicated to calculating the characteristic identity for the split Casimir operator in the defining $T_f$ and adjoint $\ad$ representations of the $osp(M|N)$ and $s\ell(M|N)$ Lie superalgebras and to constructing projectors onto invariant subspaces of the representations $T_f^{\otimes 2}$ and $\ad^{\otimes 2}$. 
 We also show that the results obtained are in full correspondence with the conclusions of \cite{IsKr} and \cite{IsPr} where analogous calculations were carried out for Lie algebras. In Section \ref{Sec5}, we write characteristic identities for the symmetric part of the split Casimir operator as well as the corresponding projectors onto symmetric invariant subspaces, uniformly (in a universal way) for both the $osp(M|N)$ and $s\ell(M|N)$ Lie algebras by using the Vogel parameters. In Section \ref{Sec6}, following the approach of \cite{Okub} and \cite{Mkr}, we find a universal form of the generating function of the higher Casimir operators of the $osp(M|N)$ and
  $s\ell(M|N)$ Lie superalgebras in the adjoint representation.

\section{General information on Lie superalgebras}\label{Sec2}

In this section, we briefly discuss the main definitions and conventions from the theory of Lie superalgebras(see, e.g., \cite{Kac}, \cite{Ber}) and introduce the notation to be used in what follows.

\subsection{Lie superalgebras and associative algebras}
A linear superspace (or $\mathbb{Z}_2$-graded space) over the field $\mathbb{C}$ is a linear space $V=V_{\oo}\oplus V_{\ol}$, which is a direct sum of the linear spaces $V_{\oo}$ and $V_{\ol}$ over the field $\mathbb{C}$. The spaces $V_{\oo}$ and $V_{\ol}$ are called even and odd, respectively. The vectors from $V$ that lie in the even subspace $V_{\oo}$ are called even, and those lying in the odd space $V_{\ol}$ are called odd. Those vectors that are either even or odd are called homogeneous. The grading of an arbitrary homogeneous vector $v\in V$ is denoted by ${\sf deg}(v) \equiv [v]\in \mathbb{Z}_2$, i.e., $[v]=0,1\mod (2)$. The linear superspace $V=V_{\oo}\oplus V_{\ol}$, where $\dim V_{\oo}=M$ and $\dim V_{\ol}=N$, will be written as $V_{(M|N)}$. The superdimension of the space $V_{(M|N)}$ is defined by $\sdim(V_{(M|N)})\equiv M-N$. Throughout the rest of this paper we always assume the basis $\{e_a\}_{a=1}^{M+N}$ of $V_{(M|N)}$ to be homogeneous, with the first $M$ of its elements being even and the last $N$ of them being odd. The grading of the basis element $e_a$ will be written as $[a]$. Note that in our convention the grading is carried by the basis vectors of the space $V_{(M|N)}$. For example, there is another (but equivalent) convention whereby the grading is carried by the coordinates 
(see, e.g., \cite{FIKK} and \cite{IKK}).

Let $\mfg$ be a Lie superalgebra over the field $\mathbb{C}$ with a Lie superbracket $\lbrack\ ,\ \rbrack:\mfg\times \mfg\to \mfg$. For arbitrary homogeneous vectors $X,Y,Z\in\mfg$ the following two properties must be satisfied (see, e.g., \cite{Kac}):
\begin{equation}\label{SLieBrac}
\begin{aligned}
\lbrack X, Y\rbrack &\in \mfg_{_{\overline{[X]+[Y]}}} \; ,\;\;\;\;\;\;
\lbrack X, Y\rbrack &= -(-1)^{[X][Y]}\lbrack Y,X\rbrack,\\
\end{aligned}
\end{equation}
\begin{equation}\label{SLieJac}
(-1)^{[X][Z]}\lbrack X,\lbrack Y,Z\rbrack\rbrack+(-1)^{[Y][X]}\lbrack Y,\lbrack Z,X\rbrack\rbrack+(-1)^{[Z][Y]}\lbrack Z,\lbrack X,Y\rbrack\rbrack=0 \;,
\end{equation}
Let $\{X_i\}$ $(i=1,...,\dim \mfg)$ be a homogeneous basis of $\mfg$. Then
\begin{equation}\label{StrConG}
\lbrack X_i,X_j\rbrack =X_k \; X^k_{\;\; ij},
\end{equation}
where the numbers $X^k_{\;\; ij}$
are the structure constants of the algebra $\mfg$. Clearly,
$X^k_{\;\; ij}=0$ as long as $([i]+[j]+[k])\mod(2)\neq 0$, and
$X^k_{\;\; ij}=-(-1)^{[i][j]}X^k_{\;\; ji}$.

Consider an associative superalgebra $\mcA=\mcA_{\oo}\oplus\mcA_{\ol}$. For any two homogeneous elements $A, B \in \mcA$ we 
 can define the bracket $\lbrack\ ,\ \rbrack$ as follows:
\begin{equation}\label{AsToLie}
\lbrack A,B\rbrack := A \, B-(-1)^{[A][B]}\, B \, A \; .
\end{equation}
It is easy to check that this bracket satisfies \eqref{SLieBrac} and \eqref{SLieJac}. Thus, the algebra $\mcA$ can be viewed as a Lie superalgebra with respect to the bracket (\ref{AsToLie}).  Following \cite{Kac}, we denote this Lie superalgebra by $(\mcA)_L$.

A representation of the Lie superalgebra $\mfg$ is a homomorphism $T:\mfg\to (\End(V_{(M|N)}))_L$, such that for any $X,Y\in\mfg$ we must have
\begin{equation}\label{RepDef}
T(\lbrack X,Y\rbrack)=\lbrack T(X),T(Y)\rbrack,
\end{equation}
where the Lie superbracket in the right hand-side of \eqref{RepDef} is defined in \eqref{AsToLie}. In this paper, the key role is played by the adjoint representation $\ad:\mfg\to (\End(\mfg))_L$, which is defined by the formula
\begin{equation}\label{AdDef}
\ad(X)\cdot Y=\lbrack X,Y\rbrack
\end{equation}
for arbitrary vectors $X,Y\in\mfg\equiv V_{\ad}$. From \eqref{StrConG} and \eqref{AdDef} it follows that the entries of the matrix 
of the operators $\ad(X_i)$ in the homogeneous basis $\{X_j\}$ are equal to the structure constants of $\mfg$:
\begin{equation}\label{AdComp}
\ad(X_i)^k{}_j=X_{\;\; ij}^k.
\end{equation}

Consider a linear superspace $V_{(M|N)}$ and an
 operator $A:V_{(M|N)}\to V_{(M|N)}$ with the matrix $||A^a{}_b||$ in some homogeneous basis $\{e_a\}_{a=1}^{M+N}$ of $V_{(M|N)}$. Recall that the supertrace of $A$ is the quantity $\str A=(-1)^{[a]}A^a{}_a$, which has the following important property:
\begin{equation}\label{strBrac}
\str(\lbrack A,B\rbrack)=0
\end{equation}
for any $A$ and $B$ acting in $V_{(M|N)}$.

The Cartan-Killing metric $\sf{g}$ of the Lie superalgebra $\mfg$ is defined in the standard way:
\begin{equation}\label{KilDefG}
{\sf g}_{ij}=\str(\ad(X_i)\ad(X_j))=(-1)^{[m]}X_{\;\; ik}^m \, X_{\;\; jm}^k.
\end{equation}
Note that $\sf{g}$ has the following properties:
\begin{align}
{\sf g}(X,Y)=(-1)^{[X][Y]}{\sf g}(Y,X)&\ & &\forall X\in\mfg_{\oalpha},\forall Y\in\mfg_{\obeta},\ \oalpha,\obeta\in \mathbb{Z}_2,\label{KilProp2}\\
{\sf g}(\lbrack X,Y\rbrack ,Z)={\sf g}(X,\lbrack Y,Z\rbrack )&\ & &\forall X,Y,Z\in \mfg,\label{KilProp3}
\end{align}
and
\begin{gather}
{\sf g}_{ij}=(-1)^{[i][j]}{\sf g}_{ji}=(-1)^{[i]}{\sf g}_{ji}=(-1)^{[j]}{\sf g}_{ji},\label{KilProp4}\\
{\sf g}_{ij}=0 \; , \quad \mbox{\rm if} \;\; [i] + [j] \neq 0 \; {\rm mod}\, (2).\label{KilProp5}
\end{gather}

In the case of the nondegenerate Cartan-Killing metric, we also introduce
 the inverse Cartan-Killing metric with the components 
 $\okg^{ij}$ given by the relations
\begin{equation}\label{InvKilG}
\okg^{ij}{\sf g}_{jk}=\delta^i_k,\qquad {\sf g}_{ij}\okg^{jk}=\delta_i^k.
\end{equation}
The metric $\okg^{ij}$ has the same properties \eqref{KilProp4} with respect to index permutations. One can use the metrics ${\sf g}_{ij}$ 
and  $\okg^{ij}$ to lower and raise the indices of the vectors and covectors in $\mfg$. Note that the raising of the indices in $\kg_{ij}$ yields $\kg^{ij}=\okg^{ik}\okg^{jm}\kg_{km}=\okg^{ji}$, so the
 metric tensor with the upper indices $\kg^{ij}$ does not coincide with the inverse matrix $\okg^{ij}$. From now on we only use $\okg^{ij}$.

Let us now introduce the structure constants of $\mfg$ with the lower indices:
\begin{equation}\label{Xlow}
X_{kij}\equiv {\sf g}_{km}X_{\;\; ij}^m \; .
\end{equation}
 From \eqref{AdComp}, \eqref{KilProp3} and \eqref{KilProp4} we deduce the following properties of $X_{ijk}$ with respect to index permutation:
\begin{equation}\label{XSym}
X_{kji}=-(-1)^{[i][j]}X_{kij} \; ,  \;\;\;
X_{jik} =-(-1)^{[k]+[j]+[k][j]}X_{kij}\; ,  \;\;\;
X_{ikj} =-(-1)^{[i][k]}X_{kij}.
\end{equation}

\subsection{The split Casimir operator and comultiplication for Lie superalgebras}
Let $\mcU(\mfg)$ denote the enveloping algebra of the Lie superalgebra $\mfg$. Consider the quadratic Casimir operator:
\begin{equation}\label{C2G}
C_2={\sf g}^{ij} \, X_i \, X_j  \; \in\mcU(\mfg) \, .
\end{equation}
In view of
 (\ref{KilProp5}), the operator
 $C_2$ is even and commutes with all generators
$X_k$ of $\mcU(\mfg)$ with respect to the bracket
\eqref{AsToLie}. Therefore, $C_2$ belongs to the centre of $\mcU(\mfg)$.

Consider two associative superalgebras $\mcA$ and $\mcB$. The graded tensor product of $\mcA$ and $\mcB$ is the associative superalgebra $\mcA\otimes \mcB$
  that coincides as a linear space with the tensor product of the spaces $\mcA$ and $\mcB$, and the multiplication in $\mcA\otimes \mcB$
  is defined for arbitrary homogeneous vectors 
  $A,A'\in \mcA$ and $B,B'\in\mcB$  as
\begin{equation}\label{AsTenMult}
(A\otimes B)\cdot (A'\otimes B')=
(-1)^{[A'][B]}AA'\otimes BB' \; .
\end{equation}
Here and below, by 'tensor product' we will always mean the 'graded tensor product', for which (\ref{AsTenMult}) holds. The matrix of the operator $A\otimes B$ in the basis $\{e_i\otimes \varepsilon_\alpha\}$ of the space $V\otimes V'$ is given by
 $$
(A\otimes B)(e_i\otimes \varepsilon_\alpha)=\left\{
\begin{array}{l}
e_k\otimes \varepsilon_\beta(A\otimes B)^{k\beta}{}_{i\alpha}\\[0.2cm]
(-1)^{[B][i]}(Ae_i)\otimes (B\varepsilon_\alpha)=(-1)^{[B][i]}(e_kA^k{}_i)\otimes (\varepsilon_\beta B^\beta{}_\alpha)
\end{array}\right.,
$$
from where for the case of a homogeneous $B$ we get:
\begin{equation}\label{TensComp}
(A\otimes B)^{k\beta}{}_{i\alpha}=(-1)^{[B][i]}A^k{}_iB^\beta{}_\alpha.
\end{equation}
The formula \eqref{TensComp} can be generalized to arbitrary (not necessarily homogeneous) operators $B\in\End(V')$ as
\begin{equation}\label{ABmat}
(A\otimes B)^{k\beta}{}_{i\alpha}=
(-1)^{([\alpha]+[\beta])[i]}A^k{}_iB^\beta{}_\alpha.
\end{equation}
Generally, for the matrix of the operator
$(A \otimes B\otimes C\otimes  \cdots \otimes E)$ 
acting in $V_1 \otimes V_2\otimes V_3\otimes  \cdots \otimes V_n$,
we have
 $$
 \begin{array}{c}
 (A \otimes B\otimes C\otimes \cdots \otimes E)^{k_1...k_n}_{\;\;\; i_1...i_n} =
 (A)^{k_1}_{\;\; i_1}(-1)^{[i_1]([k_2]+[i_2])}(B)^{k_2}_{\;\; i_2}
 (-1)^{([i_1]+[i_2])([k_3]+[i_3])}(C)^{k_3}_{\;\; i_3} \cdots \\ [0.2cm]
  \cdots (-1)^{([i_1]+[i_2]+...+[i_{n-1}])([k_n]+[i_n])}(E)^{k_n}_{\;\; i_n}\;.
 \end{array}
 $$

For an arbitrary $A:V\to V$ we define the operators $A_1,A_2:V^{\otimes 2}\to V^{\otimes 2}$ by
\begin{equation}\label{TensExp}
A_1\equiv A\otimes I,\qquad A_2\equiv I\otimes A,
\end{equation}
where $I:V\to V$ is the identity operator. Applying \eqref{TensExp}, we get:
\begin{equation*}
A_1B_2=(A\otimes I)(I\otimes B)=A\otimes B \quad \text{and} \quad B_2A_1
 =(I\otimes B)(A\otimes I)=(-1)^{[A][B]}A\otimes B,
\end{equation*}
i.e., generally,
\begin{equation*}
A_1B_2\neq B_2A_1.
\end{equation*}

The notation introduced above can be generalized to the case of any operator
 \begin{equation}
 \label{defA}
 A=\hat{A}^{i_1\dots i_r}{}_{j_1\dots j_r}e_{i_1}{}^{j_1}\otimes \dots e_{i_r}{}^{j_r} \;\; \in \;\; {\rm End}(V^{\otimes r}) \; ,
 \end{equation}
 where $e_i{}^j$ are the matrix identities which are operators that act on the space $V$ with the basis $\{ e_a \}$ as
 \begin{equation}
 \label{MatUn}
 e_i{}^j \cdot e_a = e_b \; (e_i{}^j)^b_{\;\; a} =
 e_i \; \delta^j_a \;\;\; \Leftrightarrow \;\;\;
  (e_i{}^j)^b_{\;\; a} = \delta^b_i \delta^j_a \; .
 \end{equation}
 Let $s>r$ and $1\le \alpha_1<\dots <\alpha_r\le s$.
  Define $A_{\alpha_1\dots \alpha_r} \in
   {\rm End}(V^{\otimes s})$ as
\begin{equation}\label{A123}
A_{\alpha_1\dots \alpha_r}=\hat{A}^{i_1\dots i_r}{}_{j_1\dots j_r} \;
 \underbrace{I \otimes \cdots\otimes I \otimes e_{i_1}{}^{j_1}\otimes
 I \otimes \cdots \otimes I
 \otimes  e_{i_r}{}^{j_r} \otimes I  \otimes \cdots \otimes I}_s,
\end{equation}
where each of the matrix identities $e_{i_k}{}^{j_k}$ $(k=1,\dots r)$ stands at the
$\alpha_k$-th place in the tensor product in the right-hand side of \eqref{A123},
while  at all the other places there are identity operators $I:V\to V$.
 For instance, given $A=\hat{A}^{ij}{}_{km}e_i{}^k\otimes e_j{}^m$, the operators $A_{13},A_{12},A_{23}:V^{\otimes 4}\to V^{\otimes 4}$ are defined by
\begin{equation}\label{A13}
\begin{array}{c}
A_{13}=\hat{A}^{ij}{}_{km}e_i{}^k\otimes I\otimes e_j{}^m\otimes I \; , \;\;\;
A_{12}=\hat{A}^{ij}{}_{km}(e_i{}^k\otimes e_j{}^m\otimes I\otimes I) \; ,
 \\ [0.2cm]
A_{23}=\hat{A}^{ij}{}_{km}(I\otimes e_i{}^k\otimes e_j{}^m\otimes I) \; .
\end{array}
\end{equation}

In what follows we will need the superpermutation operator
\begin{equation}\label{PgenDef}
\mcP=(-1)^{[j]}e_i{}^j\otimes e_j{}^i\qquad \implies \qquad
\mcP^{k_1k_2}{}_{m_1m_2}= (-1)^{[k_1][k_2]}\delta^{k_1}_{m_2}
\delta^{k_2}_{m_2}.
\end{equation}
Note that the operators $\mcP_{\alpha,\alpha+1}$ $(\alpha=1,\dots, s-1)$ given in accordance with \eqref{A123} 
define the representation $\tau:S_s\to \End(V^{\otimes s})$ of the symmetric group $S_s$ with generators $\sigma_\alpha$:
\begin{equation}\label{SymGr}
\begin{gathered}
\sigma_\alpha\sigma_{\alpha+1}\sigma_\alpha=\sigma_{\alpha+1}\sigma_\alpha\sigma_{\alpha+1}\qquad \forall \alpha=1,\dots, s-2,\\
\sigma_\alpha\sigma_\beta=\sigma_\beta\sigma_\alpha\qquad \forall \alpha,\beta=1,\dots,s-1,\quad |\alpha-\beta|>1,\\
\sigma_\alpha^2=e \qquad \forall \alpha=1,\dots,s-1 \; ,
\end{gathered}
\end{equation}
where  $e$ is the identity of $S_s$, and
 \begin{equation}
 \label{reprS}
\tau(\sigma_\alpha) = \mcP_{\alpha,\alpha+1} \; , \;\;\;\;
\tau(e) = I^{\otimes s} \; .
\end{equation}
Using direct calculations and the definition of $\mcP_{\alpha,\alpha+1}$,
 one can prove the following statement.
\begin{proposition}\label{SuperShift}
Let $A$
 be an operator (\ref{defA}) acting in $V^{\otimes r}$,
 and $A_{\alpha_1\dots \alpha_r}$ be an operator (\ref{A123})
 that acts on $V^{\otimes s}$ $(s >r)$.
 If $\alpha_p+1 < \alpha_{p+1}$ for $p<r$, or $\alpha_p+1\le s$ for $p=r$, then
\begin{equation}\label{prop1}
\mcP_{\alpha_p,\alpha_p+1}A_{\alpha_1\dots \alpha_p\dots \alpha_r}\mcP_{\alpha_p,\alpha_p+1}=
 A_{\alpha_1\dots \alpha_p+1\dots \alpha_r} \; .
\end{equation}
That is, if the conditions are satisfied, the superpermutation $\mcP_{\alpha_p,\alpha_p+1}$ moves the nontrivial factor at the $\alpha_p$-th position in $A_{\alpha_1\dots \alpha_r}$
to the adjacent place $\alpha_p+1$ on the right, while the identity $I$ standing at the $\alpha_p+1$th place in $A_{\alpha_1\dots \alpha_r}$ is moved to the $\alpha_p$-th position.
\end{proposition}
In particular, Statement \ref{SuperShift} implies for the operators $A_{\alpha_1\alpha_2}:V^{\otimes 4}\to V^{\otimes 4}$ defined in \eqref{A13} the relations
\begin{equation}
A_{13}=\mcP_{23}A_{12}\mcP_{23}=\mcP_{12}A_{23}\mcP_{12} \;\;
\Rightarrow \;\;
\mcP_{13}=\mcP_{23}\mcP_{12}\mcP_{23}=\mcP_{12}\mcP_{23}\mcP_{12} \; ,
\end{equation}
where the chain of equalities on the right is in accordance with \eqref{SymGr}. Besides, it follows from \eqref{prop1} that for any $A:V^{\otimes r}\to V^{\otimes r}$ we have
\begin{equation}\label{Aa1a2a3S}
A_{\alpha_1\dots\alpha_r}=\tau(\sigma)A_{1\dots r}\tau(\sigma)^{-1},
\end{equation}
where $\sigma\in S_s$, $s\ge r$ and $\tau(\sigma)$ is its image in the representation \eqref{reprS} constructed as a product of $P_{\alpha,\alpha+1}$.

For $\mcU(\mfg)$, define a homomorphic map
  $\Delta:\mcU(\mfg)\to \mcU(\mfg)\otimes \mcU(\mfg)$. It acts on the generators $X_i$ of $\mcU(\mfg)$ by
\begin{equation}
\label{coprod}
\Delta X_i=X_i\otimes I+I\otimes X_i.
\end{equation}
The map $\Delta$ is called the comultiplication of $\mcU(\mfg)$. Acting by $\Delta$ on the quadratic Casimir operator $C_2$ \eqref{C2G} yields:
\begin{equation}\label{SpCas1}
\Delta(C_2)=C_2\otimes I + I\otimes C_2+2\hC,
\end{equation}
where $\hC\in\mcU(\mfg)\otimes \mcU(\mfg)$ is called the split Casimir operator. Explicitly,
\begin{equation}\label{SpCas}
\hC=\okg^{ij}X_i\otimes X_j.
\end{equation}
 The operator $\hC$ has the property of $\ad$-invariance,
 i.e., for all generators $X_i \in \mcU(\mfg)$ we have, according to \eqref{SpCas1}:
\begin{equation}\label{CasAdInv}
\lbrack \hC,\Delta X_i\rbrack =\frac{1}{2}\Delta(\lbrack C_2,X_i\rbrack)-\frac{1}{2}\lbrack C_2,X_i\rbrack \otimes I-\frac{1}{2}I\otimes \lbrack C_2,X_i\rbrack=0 \; .
\end{equation}

It is convenient to use the notion of tensor product of the enveloping superalgebras ${\cal U}(\mathfrak{g})$ 
and comultiplication in ${\cal U}(\mathfrak{g})$
to define the tensor product $T\otimes T'$ of the representations $T:\mfg\to(\End(V))_L$ and $T':\mfg\to(\End(V'))_L$ of the Lie superalgebra $\mfg$. For an arbitrary homogeneous vector $X\in\mfg$ we define
$(T\otimes T')(X)$ as
\begin{equation}\label{TenRepDef}
(T\otimes T')(X)\equiv (T\otimes T') \big(\Delta (X)\big)
=T(X)\otimes T'(I)+T(I)\otimes T'(X) \; ,
\end{equation}
so for any homogeneous vectors $v\in V$ and $u\in V'$
we have
\begin{equation}
(T\otimes T')(X)\cdot (v\otimes u)=(T(X)\cdot v)\otimes u+(-1)^{[X][v]}v\otimes (T'(X)\cdot u).
\end{equation}
 The map thus defined
$(T\otimes T'): \;\; \mcU(\mfg)\to \End(V\otimes V')$
is indeed homomorphic and therefore is a representation of $\mcU(\mfg)$ on
 $V\otimes V'$.

 From (\ref{CasAdInv}) and \eqref{TenRepDef} one can infer that for any represenations $T$ and $T'$ the operator $(T\otimes T')(\hC)$ commutes with $(T\otimes T')(X)$ for all $X \in \mathfrak{g}$.
 Recall that by Schur's lemma (more precisely, by its generalization to the case of Lie superalgebras)
 for each {\it irreducible} representation $\widetilde{T}$
 of any complex Lie superalgebra ${\cal A}$, an even operator
 $A$ that commutes with all the elements of ${\cal A}$ in the representation $\widetilde{T}$ must be proportional to the identity operator, that is $A=\lambda I$, where $\lambda\in \mathbb{C}$.
 Thus, if an irreducible representation $\widetilde{T}$ of $\mcU(\mfg)$ is contained in $(T\otimes T')$, we have
 $\widetilde{T}(\hC) \sim I_{\widetilde{T}}$.
The following corollary of Schur's lemma is central for our work: if a representation $T\otimes T$ of a Lie superalgebra $\mfg$ in the space $V_T\otimes V_T$ is completely reducible, then $V_T\otimes V_T$ can be expanded as a direct sum of invariant eigenspaces of the operator $(T\otimes T)(\hC)$. From now on, we denote this operator by $\hC_T$. If the operator $\hC_T$ satisfies a characteristic identity
\begin{equation}
\label{charct}
(\hC_T-a_1 I_T^{\otimes 2})(\hC_T-a_2 I_T^{\otimes 2})\dots
(\hC_T-a_p I_T^{\otimes 2})=0,
\end{equation}
where $I_T^{\otimes 2}$ is the identity operator on
$V_T\otimes V_T$, all the complex numbers $a_1,\ a_2,\dots,a_p$ are different, and crossing out any of the parentheses on the left of (\ref{charct})
breaks the identity, then the numbers $a_1,\ a_2,\dots,a_p$ are the eigenvalues of $\hC_T$, and a projector onto the eigenspace of $\hC_T$ corresponding to the eigenvalue $a_j$ is given by the formula
\begin{equation}\label{GenPr}
\proj_j\equiv\proj_{a_j}=
\prod_{\substack{i=1\\i\neq j}}^{p}
\frac{\hC_T-a_i I_T^{\otimes 2}}{a_j-a_i} \; .
\end{equation}
Moreover, $\proj_k \proj_j = \delta_{kj} \proj_j$ and
$\sum_{i=1}^p \proj_i = I_T^{\otimes 2}$. We emphasize that the spaces extracted by the projectors $\proj_j$ are not necessarily spaces of irreducible representations of $\mfg$, as there may exist other nontrivial invariant operators in $V_T\otimes V_T$
 that cannot be expressed as polynomials in $\hC_T$. Therefore, the spaces $\proj_j(V_T\otimes V_T)$
can in principle be further expanded into a direct sum of nontrivial invariant subspaces.

Unlike the case of Lie algebras, reducible representations of simple Lie superalgebras are not always completely reducible. Correspondingly, the Casimir operators in such representations are not always diagonilizable: this situation is present in our paper. The operator $\hC_T$ is not diagonilizable if and only if it does not satisfy any identity of the form \eqref{charct} with pairwise different $a_i$. In this case, $\hC_T$ must satisfy
\begin{equation}\label{charct2}
(\hC_T-a_1 I_T^{\otimes 2})^{k_1}(\hC_T-a_2 I_T^{\otimes 2})^{k_2}\dots
(\hC_T-a_p I_T^{\otimes 2})^{k_p}=0,
\end{equation}
where all $k_i\in\mathbb{Z}_{\ge 1}$ are minimal, i.e., subtracting $1$ from any of them breaks the identity and, as earlier, $a_i$ are pairwise different. Instead of projectors onto eigenspaces of $\hC_T$, we can construct projectors onto its generalised eigenspaces (weight spaces):\footnote{Here by weight spaces we mean spaces in which $\hC_T$ acts as a Jordan cell with the corresponding eigenvalue}:
\begin{equation}\label{GenPr2}
\proj_j\equiv \proj_{a_j}=I^{\otimes 2}_T-\Bigg(I^{\otimes 2}_T-\prod_{\substack{i=1\\i\neq j}}^p\bigg(\frac{\hC_T-a_i I^{\otimes 2}_T}{a_j-a_i}\bigg)^{k_i}\Bigg)^{k_j}.
\end{equation}
If $k_i=1$ for some $i$, then the image of $\proj_i$ is an eigenspace of $\hC_T$. If $k_i>1$, then $\proj_i$ projects onto a generalized eigenspace. Note that for $k_1=k_2=\dots=k_p=1$ \eqref{charct2} and \eqref{GenPr2} turn into \eqref{charct} and \eqref{GenPr}, respectively.

\subsection{The split Casimir operator for simple complex Lie superalgebras with nondegenerate Cartan-Killing metric}\label{21}
 Hereinafter, we consider only those simple Lie superalgebras, for which the Cartan-Killing metric ${\sf g}_{ab}$ is nondegenerate, i.e., there must be an inverse metric $\okg^{ab}$ satisfying (\ref{InvKilG}). It also implies that the structure constants $X^k_{\;\; ij}$ satisfy the relation
 \be
 \lb{strX}
 {\rm str} (\ad(X_i)) = (-1)^k \; X^k_{\;\; ik} = 0 \; , \;\;\;\;\;
 \forall i \; .
 \ee

Using \eqref{TensComp}, \eqref{AdComp} and \eqref{SpCas}, we can find the components of the split Casimir operator in the adjoint representation with respect to the basis $\{X_a\otimes X_b\}$ of $\Vad\otimes \Vad$ ($\Vad$ coincides with $\mfg$ as a vector space and is defined to be the space of the adjoint representation):
\begin{equation}\label{CadGenC}
(\hCad)^{i_1i_2}{}_{j_1j_2}=(-1)^{[a_2][j_1]}
\okg^{a_1a_2} \, X^{i_1}{}_{a_1j_1} \, X^{i_2}{}_{a_2j_2}.
\end{equation}
Here we need other three ad-invariant operators in $\Vad\otimes \Vad$ with the components
\begin{equation}\label{IPKGenC}
(\bI)^{i_1i_2}{}_{j_1j_2}=\delta^{i_1}_{j_1}\delta^{i_2}_{j_2},\qquad (\bP)^{i_1i_2}{}_{j_1j_2}=(-1)^{[i_1][i_2]}\delta^{i_1}_{j_2}\delta^{i_2}_{j_1},\qquad (\bK)^{i_1i_2}{}_{j_1j_2}=\okg^{i_1i_2}\kg_{j_1j_2}.
\end{equation}
Therefore, $\bP$ is the operator of superpermutation of the two spaces $\Vad$.  One can easily check that for any $X\in \mfg$,
\begin{equation}\label{XPcom}
[\ad^{\otimes 2} (\Delta \; X),\bP]=0,
\end{equation}
As $\bP$ is invariant under the adjoint action of $\mfg$, then so are its eigenspaces, which are extracted by the projectors $\frac{1}{2}(\bI+\bP)$ and $\frac{1}{2}(\bI-\bP)$.

Using the component form \eqref{IPKGenC} and \eqref{CadGenC} of $\bI$, $\bP$, $\bK$ and $\hCad$, one can check the identities
\begin{equation}\label{IPKalg}
\bP\cdot \bP=\bP,\qquad \bP\hCad\bP=\hCad,\qquad \bP\cdot \bK=\bK\cdot \bP=\bK,
\end{equation}
\begin{equation}\label{KKK}
\bK\cdot \bK=\sdim \mfg\cdot \bK.
\end{equation}

Applying the projectors
\begin{equation}
\bPad_\pm=\frac{1}{2}(\bI\pm\bP)
\end{equation}
to $\hCad$, we define the symmetric and antisymmetric parts of the split Casimir operator:
\begin{equation}\label{CpmDef}
\hC_{\pm}=\bPad_{\pm}\hCad=\hCad\bPad_{\pm}
\end{equation}
where the last equality is verified by using \eqref{IPKalg}. From \eqref{CpmDef} and the relations
\begin{equation}
\bPad_++\bPad_-=\bI,\quad \bPad_+\bPad_-=\bPad_-\bPad_+=0
\end{equation}
that follow from \eqref{CpmDef}, we instantly get:
\begin{equation}\label{CpmMult}
\hC_{\ad}=\hCp+\hCm,\qquad \hCp\hCm=\hCm\hCp=0.
\end{equation}

Utilizing the symmetry properties \eqref{XSym} of the structure constants of $\mfg$ with respect to index permutation, the graded Jacobi identity \eqref{SLieJac}, and the identity
\begin{equation}\label{XXdelta}
(-1)^{[i_1][i_2]}X^{i_1i_2}{}_aX^b{}_{i_1i_2}=-\delta^b_a \;\;\;
\Leftrightarrow \;\;\;
(-1)^{[i_1][i_2]} X^{i_1i_2}{}_a \, X_{i_1i_2b}=-{\sf g}_{ab} \; ,
\end{equation}
where $X^{i_1i_2}{}_a=\kg^{i_2j_2}X^{i_1}{}_{j_2a}$, we get a convenient form of $\hCm$:
\begin{equation}\label{CmconvC}
(\hCm)^{i_1i_2}{}_{j_1j_2}=\frac{1}{2}(-1)^{[i_1][i_2]}X^{i_1i_2}{}_aX^a{}_{j_1j_2}.
\end{equation}

By \eqref{CmconvC}, \eqref{XSym}, \eqref{KilDefG} and identities \eqref{XXdelta}, one can derive the following equality:
\begin{equation}\label{Cm2byCm}
\hCm^2=-\frac{1}{2}\hCm.
\end{equation}
Besides, from \eqref{CmconvC}, \eqref{XSym}, \eqref{KilProp4} we get:
\begin{equation}\label{CadpmKRel}
\hCad\bK=\bK\hCad=-\bK,\qquad \hCm\bK=\bK\hCm=0,\qquad \hCp\bK=\bK\hCp=-\bK,
\end{equation}
where the last equality in \eqref{CadpmKRel} is a consequence of the other two and of the first relation in \eqref{CpmMult}.

Supertraces of the operators $\bI$, $\bP$, $\bK$, $\hCad$, $\hCp$, and $\hCm$, as well as of some of their powers can be obtained by using \eqref{strX}, \eqref{CadGenC}, \eqref{IPKGenC}, \eqref{CmconvC}, \eqref{Cm2byCm}:
\begin{equation}\label{strGenMany}
\begin{gathered}
\bstr (\hCad)=0,\qquad \bstr(\hC_{\pm})=\pm\frac{1}{2}\sdim\mfg,\qquad \bstr (\hCad^2)=\sdim\mfg,\\
\bstr(\hCm^2)=-\frac{1}{2}\bstr(\hCm)=\frac{1}{4}\sdim\mfg,\qquad \bstr(\hCp^2)=\bstr(\hCad^2-\hCm^2)=\frac{3}{4}\sdim \mfg,\\
\bstr(\hCm^3)=\frac{1}{4}\bstr(\hCm)=-\frac{1}{8}\sdim\mfg,\qquad \bstr(\hCp^3)=\bstr(\hCad^3-\hCm^3)=-\frac{1}{8}\sdim \mfg,\\
\bstr (\bK)=\sdim\mfg,\qquad \bstr (\bI)=(\sdim\mfg)^2,\qquad \bstr (\bP)=\sdim \mfg.
\end{gathered}
\end{equation}
Here $\bstr=\str_1\str_2$ is the trace in $V_{\ad}\otimes V_{\ad}$, and the indices $1$ and $2$ in $\str_1$ and $\str_2$ refer to the tensor components of $V_{\ad}\otimes V_{\ad}$.

\section{The $osp(M|N)$ Lie superalgebra}\label{Sec3}

 There exist various conventions on how to define the orthosymplectic Lie superalgebras $osp(M|N)$ (see e.g.\cite{Kac} --
 \cite{Sorba2}), all of which are equivalent.
 In this section, we fix the definition of these algebras
 that was formulated in \cite{FIKK}, \cite{IKK}, as it is the most convenient for our purpose.
 For this, we introduce a scalar product $\varepsilon$ in the space  $\varepsilon$ $V_{(M|N)}$ (here $N=2n$ is even). Its components $\varepsilon_{ab}\equiv \varepsilon(e_a,e_b)$ in some homogeneous basis $\{e_a\}_{a=1}^{M+N}$ of $V_{(M|N)}$ are given by the matrix:
\begin{equation}\label{Jay}
\varepsilon=
\left(
\begin{array}{c|c}
I_M & 0\\
\hline
0 & J_{N}\\
\end{array}
\right).
\end{equation}
Here $I_M$ is the $(M\times M)$ identity matrix. The antisymmetric matrix $J_N$ is
\begin{equation}
J_N=
\left(
\begin{array}{c c}
0 & I_n\\
-I_n & 0
\end{array}
\right).
\end{equation}
The definition \eqref{Jay} of $\varepsilon$ implies the following relations on its components:
\begin{equation}\label{epsProp}
\varepsilon _{ba}=(-1)^{[a]}\varepsilon _{ab}=(-1)^{[b]}\varepsilon _{ab}=(-1)^{[a][b]}\varepsilon _{ab}\;.
\end{equation}
The metric $\varepsilon$ can be used to raise and lower indices:
\begin{equation}\label{UpLow}
z_{a\dots}=\varepsilon _{ab}z^b{}_{\dots},\qquad z^{a\dots}=\ovarepsilon ^{ab}z_{b}{}^{\dots}\;,
\end{equation}
where $\ovarepsilon ^{ab}$ are the components of the matrix $\varepsilon^{-1}$:
\begin{equation}
\varepsilon ^{ab}\varepsilon _{bc}=\delta^a_c,\qquad \varepsilon _{ab}\varepsilon ^{bc}=\delta_a^c\;.
\end{equation}
By \eqref{UpLow}, $
\varepsilon^{ab}=\ovarepsilon^{ac}\ovarepsilon^{bd}\varepsilon_{cd}=\ovarepsilon^{ba}$, i.e., the matrix of the metric tensor with the upper indices $\varepsilon^{ab}$ does not coincide with the inverse matrix $\ovarepsilon^{ab}$. From now on, we only use the matrix $\ovarepsilon^{ab}$.

Define the $osp(M|N)$ Lie superalgebra as the algebra of operators $A:V_{(M|N)}\to V_{(M|N)}$ that leave the metric $\varepsilon$ invariant. For homogeneous $A$ this means:
\begin{equation}\label{epsInv1}
\varepsilon(A(u),v)+(-1)^{[A][u]}\varepsilon(u,A(v))=0,
\end{equation}
for any homogeneous $u,v\in V_{(M|N)}$. If $A=A_{\oo}+A_{\ol}$ is not homogeneous, with $A_{\oo}$ and $A_{\ol}$ being its even and odd parts, respectively, we require equation \eqref{epsInv1} to hold simultaneously for $A_{\oo}$ and $A_{\ol}$. Expanding $u=u^a e_a$ and $v=v^b e_b$ over the basis $\{e_a\}$, using the definitions of the components $\varepsilon_{ab}=\varepsilon(e_a,e_b)$ of the metric $\varepsilon$ and of the operator $A(e_a)=e_b A^b{}_a$, as well as the fact that the grading of a homogeneous operator $A$ with nonzero components $A^b{}_a$ equals $[a]+[b]$, we can rewrite \eqref{epsInv1} in the component form \cite{FIKK}, \cite{IKK}:
\begin{equation}\label{epsInv2}
A^c{}_a\varepsilon_{cb}+(-1)^{[a]+[a][b]}\varepsilon_{ac}A^c{}_b=0.
\end{equation}
Note that for the left-hand side of \eqref{epsInv2} to be well-defined, $A$ needs not be homogeneous; thus, \eqref{epsInv2} can be viewed as the definition of invariance of a metric with respect to the action of an arbitrary operator $A$. Multiplying \eqref{epsInv2} by $(-1)^{[a]+[a][b]}$ yields:
\begin{equation}\label{epsInv3}
(-1)^{[a]+[a][b]}A^c{}_a\varepsilon_{cb}+\varepsilon_{ac}A^c{}_b=0.
\end{equation}
To draw an analogy between \eqref{epsInv3} and the definition of the matrix Lie algebras $so(M)$ and $sp(N)$, we introduce the operation of supertransposition of any $A:V_{(M|N)}\to V_{(M|N)}$, the application of which results in an operator $A^T:\overline{V}_{(M|N)}\to\overline{V}_{(M|N)}$ where $\overline{V}_{(M|N)}$ is the dual space of $V_{(M|N)}$. The components $(A^T)_b{}^a$ of $A^T$, which are defined by$A^T(\epsilon^a)=\epsilon^b(A^T)_b{}^a$ with $\{\epsilon^a\}$ being the dual basis of $\{e_a\}$ in $\overline{V}_{(M|N)}$, are given explicitly by
\begin{equation}\label{transpC}
(A^T)_a{}^b=(-1)^{[a]+[a][b]}A^b{}_a.
\end{equation}
Thus, for the $(M+N)\times (M+N)$ matrix of $A^T$ we have:
\begin{equation}\label{transp}
A=\left(\begin{array}{c|c}
X & Y\\
\hline
Z & W
\end{array}
\right)\implies
A^T
=
\left(
\begin{array}{c|c}
X^t & Z^t\\
\hline
-Y^t & W^t
\end{array}
\right).
\end{equation}
Here $X,\ Y,\ Z,\ W$ are $M\times M$, $M\times N$, $N\times M$, $N\times N$ matrices, respectively, and $t$ denotes the usual matrix transposition.

Using \eqref{transpC} and \eqref{epsProp}, we can rewrite \eqref{epsInv3} as
\begin{equation}\label{epsInv4}
(A^T)_a{}^c \, \varepsilon_{cb}+\varepsilon_{ac} \, A^c{}_b=0.
\end{equation}
The matrix form of \eqref{epsInv4} is:
\begin{equation}\label{epsInv5}
A^T\varepsilon+\varepsilon A=0.
\end{equation}
The form of \eqref{epsInv5} coincides with that of an analogous expression used in the definitions of the matrix algebras $so(M)$ and $sp(N)$
$(N=2n)$ with the appropriate choice of the metric $\varepsilon$ and supertransposition substituted with the regular transposition. Using \eqref{transp}, where $Y,Z$ and $W$ are viewed as block matrices, and \eqref{epsInv5} we infer the explicit form of the matrix of $A$:
\begin{equation}\label{ospMat}
A=\left(
\begin{array}{c|c c}
X & Q & S\\
\hline
-S^t & E & F\\
Q^t & G & -E^t
\end{array}
\right)\;.
\end{equation}
Here $X$ is a matrix of dimension $M\times M$ (as in (\ref{transp})),
$M\times n$ blocks $Q$ and $S$ form the $M\times 2n$ matrix $Y=(Q\; S)$,
and the $n\times n$ matrices $E$, $F$, $G$ comprise $W$. Furthermore, $X$, $F$ and $G$ satisfy $X^t=-X$, $F^t=F$ and $G^t=G$.

Note that supertransposition does not possess some properties intrinsic to the usual transposition. In general, for some $A,B:V_{(M|N)}\to V_{(M|N)}$ it may be that $
(A^T)^T\neq A$ and $(AB)^T\neq B^T A^T.
 $
Nevertheless,
 by \eqref{transpC} and \eqref{transp}:
\begin{equation}\label{ospTrel}
[A,B]^T=-[A^T,B^T] \; ,
\end{equation}
 from where it follows that the vector space of all operators $A$, satisfying
  \eqref{epsInv4} and \eqref{epsInv5}, is closed under the Lie bracket \eqref{AsToLie}. Therefore, it forms a Lie superalgebra.

To find all solutions of  \eqref{epsInv2}, we introduce the following operators in the space $\End(V_{(M|N)})$:
\begin{equation}\label{projPM}
\proj_\pm(E)^a{}_c=\frac{1}{2}\big(E^a{}_c\pm(-1)^{[c]+[c][a]}
\varepsilon_{cb}E^b{}_d\ovarepsilon^{da}\big)\;,
\end{equation}
where $E\in\End(V_{(M|N)})$, and $E^a{}_b$ is its matrix. It is easy to check that the operators $\proj_\pm$ satisfy
\begin{equation}\label{ProjProds}
\proj_A \; \proj_B =\proj_A \; \delta_{AB} \qquad (A,B=+,-) ,\qquad
\proj_++\proj_-=I_{M+N} \; .
\end{equation}
and therefore constitute a full system of mutually orthogonal projectors in $\End(V_{(M|N)})$.

In terms of $\proj_\pm$, equation \eqref{epsInv2} is rewritten as:
\begin{equation}\label{GenCon3}
\big(\proj_+A\big)^a{}_b=0\iff \proj_+A=0\iff \proj_- A=A\;.
\end{equation}
The latter condition is satisfied if and only if $A$ lies in the image of $\proj_-$. Hence, the matrix of any operator $A\in osp(M|N)$
is of the form:
\begin{equation}\label{GenCon4}
A^a{}_b=E^a{}_b-(-1)^{[b]+[b][a]}\varepsilon _{bc}E^c{}_d\ovarepsilon ^{da},
\end{equation}
where $||E^a{}_b||\in \Mat_{M+N}(\mathbb{C})$ is an arbitrary matrix.

 The basis elements $M_i{}^j \in osp(M|N)$
 in the defining representation are realized as matrices
 $(M_i{}^j)^a{}_b$ obtained from \eqref{GenCon4}
 by the substitution
$E^a{}_b \to (e_i{}^j)^a{}_b$:
\begin{equation}\label{Bas1}
(M_i{}^j)^a{}_b=(e_i{}^j)^a{}_b-(-1)^{[b]+[b][a]}
\varepsilon _{bc}(e_i{}^j)^c{}_d\ovarepsilon ^{da}\;,
\end{equation}
 where $e_i{}^j$ are the matrix identities in $V_{(M|N)}$, $(e_i{}^j)^a{}_b=\delta^j_b\delta^a_i$,
 see (\ref{MatUn}).
Lowering the index $j$ in \eqref{Bas1} via the metric $\varepsilon$ given in \eqref{Jay}, we get the final form of the matrices of the $osp(M|N)$ Lie superalgebra basis elements $M_{ij}$:
\begin{equation}\label{MatBas}
(M_{ij})^a{}_b=\varepsilon _{jb}\delta^a_i-(-1)^{[a][b]}\varepsilon _{ib}\delta^a_j=\varepsilon_{jb}\delta^a_i-(-1)^{[i][j]}\varepsilon _{ib}\delta^a_j\;.
\end{equation}
The degree of $M_{ij}$ is $[i]+[j]$. Moreover, $M_{ij}$ satisfy:
\begin{equation}
M_{ij}=-(-1)^{[i][j]}M_{ji}.
\end{equation}
Taking this condition into account, we require the components of any vector (or covector) from the $osp(M|N)$ Lie superalgebra to satisfy the same index permutation symmetry, that is:
$
X^{ij}=-(-1)^{[i][j]}X^{ji}$, $Y_{ij}=-(-1)^{[i][j]}Y_{ji}$,
where $X^{ij}$ are the coordinates of an arbitrary vector from $osp(M|N)$, and $Y_{ij}$ are the coordinates of an arbitrary covector. This requirement allows us to achieve uniqueness in assignment of coordinates to vectors of $osp(M|N)$ in the basis\eqref{MatBas} (and to covectors in the dual basis of \eqref{MatBas}).

The Lie superbracket \eqref{AsToLie} of $M_{ij}$ is
\begin{equation}
\lbrack M_{ij},M_{km}\rbrack=\varepsilon _{jk}M_{im}-(-1)^{[k][m]}\varepsilon _{jm}M_{ik}-(-1)^{[i][j]}\varepsilon _{ik}M_{jm}+(-1)^{[i][j]+[k][m]}\varepsilon _{im}M_{jk}
\end{equation}
The structure constants $X_{i_1i_2,j_1j_2}{}^{k_1k_2}$ of $osp(M|N)$ are defined by
\begin{equation}
\lbrack M_{i_1i_2},M_{j_1j_2}\rbrack = X_{i_1i_2,j_1j_2}{}^{k_1k_2} M_{k_1k_2}\;,
\end{equation}
and have the following explicit form:
\begin{equation}\label{OspStrCon}
\begin{aligned}
X^{k_1k_2}{}_{i_1i_2,j_1j_2}&=\varepsilon _{i_2j_1}\delta_{i_1}^{(k_1}\delta_{j_2}^{k_2)}-(-1)^{[j_1][j_2]}\varepsilon _{i_2j_2}\delta_{i_1}^{(k_1}\delta_{j_1}^{k_2)}-\\
&-(-1)^{[i_1][i_2]}\varepsilon _{i_1j_1}\delta_{i_2}^{(k_1}\delta_{j_2}^{k_2)}+(-1)^{[i_1][i_2]+[j_1][j_2]}\varepsilon _{i_1j_2}\delta_{i_2}^{(k_1}\delta_{j_1}^{k_2)}\;,
\end{aligned}
\end{equation}
where $A^{(k_1k_2)}$ stands for:
\begin{equation}\label{IndASym}
A^{(k_1k_2)}=\frac{1}{2}(A^{k_1k_2}-(-1)^{[k_1][k_2]}A^{k_2k_1}).
\end{equation}

By \eqref{KilDefG}, the Cartan-Killing metric \eqref{KilDefG} of $osp(M|N)$ in the basis \eqref{MatBas} equals
\begin{equation}\label{Kil}
\begin{aligned}
{\sf g}_{i_1i_2,j_1j_2}&=(-1)^{[m_1]+[m_2]}X^{m_1m_2}{}_{i_1i_2,k_1k_2}
X^{k_1k_2}{}_{j_1j_2,m_1m_2}\\
&=2\big(\omega-2\big)\lbrack \varepsilon _{i_1j_2}\varepsilon _{i_2j_1}-(-1)^{[j_1][j_2]}\varepsilon _{i_1j_1}\varepsilon _{i_2j_2}\rbrack \;,
\end{aligned}
\end{equation}
where $\omega \equiv M-N$. For $\omega=2$ the metric \eqref{Kil} is degenerate and thus this case is omitted in what follows. The identity operator $\hI$ that acts on the algebra $osp(M|N)$ (which is considered here as a vector space embedded into $V_{(M|N)}^{\otimes 2}$) has the following components in the basis \eqref{MatBas}
\begin{equation}
\hI^{i_1i_2}{}_{j_1j_2}=\frac{1}{2}\big(\delta^{i_1}_{j_1}
\delta^{i_2}_{j_2}-(-1)^{[i_1][i_2]}\delta^{i_1}_{j_2}
\delta^{i_2}_{j_1}\big)\; ,
\end{equation}
 and appears to be a projector onto (super)antisymmetric second-rank tensors (see Section {\bf \ref{projt2}}).
The components of the inverse Cartan-Killing metric are defined by \eqref{InvKilG}:
\begin{equation}
{\sf g}^{i_1i_2,j_1j_2}{\sf g}_{j_1j_2,k_1k_2}=\hI^{i_1i_2}{}_{k_1k_2}, \qquad {\sf g}_{i_1i_2,j_1j_2}{\sf g}^{j_1j_2,k_1k_2}=\hI^{k_1k_2}{}_{i_1i_2}.
\end{equation}
Direct calculations yield their explicit form:
\begin{equation}\label{InvKil}
{\sf g}^{i_1i_2,j_1j_2}=\frac{1}{8(\omega-2)}\big(\ovarepsilon^{i_1j_2}
\ovarepsilon^{i_2j_1}-(-1)^{[i_1][i_2]}\ovarepsilon^{i_1j_1}
\ovarepsilon^{i_2j_2}\big)\;.
\end{equation}

\subsection{Projectors onto invariant subspaces of tensor product of two defining representations\label{projt2}}
Using \eqref{MatBas}, \eqref{InvKil} and \eqref{SpCas}, we write the split Casimir operator in the tensor product of two defining representations
\cite{FIKK}, \cite{IKK}:
\begin{equation}
\hC_f^{k_1k_2}{}_{m_1m_2}=\frac{1}{2(\omega-2)}\big((-1)^{[k_1][k_2]}
\delta^{k_1}_{m_2}\delta^{k_2}_{m_1}-\ovarepsilon ^{k_1k_2}
\varepsilon _{m_1m_2}\big)\;.
\end{equation}
Define the operators $\bii,\mcP,\mcK:V^{\otimes 2}_{(M|N)}\to V^{\otimes 2}_{(M|N)}$:
 (the matrix forms of these operators are given on the right of these formulas):
\begin{align}
\bii&=e_i{}^i\otimes e_j{}^j &\implies & &
\bii^{k_1k_2}{}_{m_1m_2}= \delta^{k_1}_{m_1}\delta^{k_2}_{m_2} \;,\label{Idef}\\
\mcP&=(-1)^{[j]}e_i{}^j\otimes e_j{}^i &\implies & &
\mcP^{k_1k_2}{}_{m_1m_2}= (-1)^{[k_1][k_2]}\delta^{k_1}_{m_2}
\delta^{k_2}_{m_2} \;,\label{Pdef}\\
\mcK&=(-1)^{[i]+[i][k]}\varepsilon^{ij}
\varepsilon_{km}e_i{}^k\otimes e_j{}^m
&\implies & &
\mcK^{k_1k_2}{}_{m_1m_2}=  \varepsilon^{k_1k_2}\varepsilon_{m_1m_2}   \; .
\label{Kdef}
\end{align}
Here $\bii$ is the identity operator,
$\mcP$ is the superpermutation, and $e_i{}^j$ are the defined in \eqref{MatUn} matrix identities that act on $V_{(M|N)}$.
$\bii$, $\mcP$ and $\mcK$ have the following properties:
\begin{equation}
\mcP^2=\bii,\quad \mcK^2=\omega \mcK,\quad \mcP\mcK=\mcK\mcP=\mcK\;.
\end{equation}
In terms of $\mcP$ and $\mcK$, the split Casimir operator in the defining representation can be written as \cite{FIKK}, \cite{IKK}
\begin{equation}\label{hCf}
\hC_f=\frac{1}{2(\omega-2)}\big(\mcP-\mcK\big)\;.
\end{equation}
The characteristic identity for $\hC_f$ has degree three:
\begin{equation}\label{fChPlUn}
\hC_f^3+\frac{\omega-1}{2(\omega-2)}\hC_f^2-
\frac{1}{4(\omega-2)^2}\hC_f-\frac{\omega-1}{8(\omega-2)^3} = 0 \;.
\end{equation}
Formula \eqref{fChPlUn} is easily verified by using (\ref{hCf}) and the relations
 \begin{equation}\label{hCf23}
\hC_f^2=\frac{1}{4(\omega-2)^2}\bii+\frac{1}{4(\omega-2)}\mcK\;,\qquad\quad \hC_f^3=\frac{1}{8(\omega-2)^3}\big\lbrack \mcP-(\omega^2-3\omega + 3)\mcK\big\rbrack\;.
\end{equation}
The left-hand side of \eqref{fChPlUn} can be factorized, which results in:
\begin{equation}\label{fChPl}
\Big(\hC_f-\frac{1}{2(\omega-2)}\Big)
\Big(\hC_f+\frac{1}{2(\omega-2)}\Big)
\Big(\hC_f+\frac{\omega-1}{2(\omega-2)}\Big)=0 \; .
\end{equation}
Using this factorized form of the characteristic identity for $\hC_f$ and \eqref{GenPr},
 where we put $p=3$, fix the roots
$$
a_1=\frac{1}{2(\omega-2)} \; , \;\;\;\; a_2=-\frac{1}{2(\omega-2)}\; , \;\;\;\; a_3=\frac{1-\omega}{2(\omega-2)} \; ,
$$
and utilize \eqref{hCf}, and the left hand-side of \eqref{hCf23} for $\hC_f$ and $\hC_f^2$ gives us three projectors onto invariant subspaces of $V_{(M|N)}^{\;\; \otimes 2}$:
\begin{equation}\label{ospFproj}
\proj_1 =\frac{1}{2}(\bii+\mcP)-\frac{1}{\omega}\mcK\;,\quad
\proj_2 =\frac{1}{2}(\bii-\mcP)\;,\quad
\proj_3 =\frac{1}{\omega}\mcK\;.
\end{equation}
Note that the substitutions $\omega\to M$ and $\omega\to -N$
turn the projectors \eqref{ospFproj} into the corresponding projectors onto invariant subspaces of the representation $T_f^{\otimes 2}$ of the $so(M)$ and $sp(N)$ Lie algebras  (for the explicit formulas see, e.g. \cite{Cvit} and \cite{Book2}).

 To conclude this subsection, we show that the solution $R^{i_1 i_2}_{\;\; k_1 k_2}(u)$
 of the graded Yang-Baxter equation \cite{KulSkl}
(see also \cite{FIKK}, \cite{Isaev}):
\begin{equation}
\label{grYBE}
 \begin{array}{c}
 R^{i_{1}i_{2}}_{\;\; j_{1}j_{2}}(u) (-1)^{[j_1][j_2]}
R^{j_{1}i_{3}}_{\;\; k_{1}j_{3}}(u+v) (-1)^{[k_1][j_2]}
R^{j_{2}j_{3}}_{\;\; k_{2}k_{3}}(v) = \\ [0.2cm]
= R^{i_{2}i_{3}}_{\;\; j_{2}j_{3}}(v) (-1)^{[i_1][j_2]}
R^{i_{1}j_{3}}_{\;\; j_{1}k_{3}}(u+v) (-1)^{[j_1][j_2]}
R^{j_{1}j_{2}}_{\;\; k_{1}k_{2}}(u),
\end{array}
\end{equation}
which is invariant under the action of $osp(M|N)$ in the defining representation, can be written in terms of the split Casimir operator (\ref{hCf}).  Recall that this solution can be written in several equaivalent ways \cite{FIKK}, \cite{IKK}, \cite{Isaev}:
\begin{equation}\label{ospYBsol}
R(u)=\frac{1}{1-u}\Big(u+\mcP-\frac{u}{u+\omega/2-1}\mcK\Big)=
 \frac{1+u}{1-u}\proj_1-\proj_2+\frac{\omega/2-1-u}{\omega/2-1+u}\proj_3.
\end{equation}
 The important point to note here is that this solution can be written as a rational function of $\hC_f$:
 \be\lb{best5}
 R(u) =\frac{(\omega-2)\hC_f +1/2 +u}{(\omega-2)\hC_f  +1/2 -u} \; .
 \ee
 This form of the $R$-matrix generalizes that obtained in \cite{IsKr}
 for the $so(M)$ and $sp(2n)$ Lie superalgebras in the defining representation.
Note that the solution \eqref{ospYBsol}, (\ref{best5}) is unitary, i.e., $\mcP R(u)\mcP R(-u)=\bii$.

\subsection{Projectors onto invariant subspaces of the tensor product of two adjoint representations}
In order to find a characteristic identity for the split Casimir operator in the adjoint representation $\hC_{\ad}$, we first write the components of the basis elements of $osp(M|N)$ in this representation. By \eqref{AdComp}, they coincide with the structure constants \eqref{OspStrCon} of $osp(M|N)$:
\begin{equation}\label{AdBas}
\begin{aligned}
(M_{i_1i_2})^{k_1k_2}{}_{j_1j_2}&=\varepsilon _{i_2j_1}\delta_{i_1}^{(k_1}\delta_{j_2}^{k_2)}-(-1)^{[j_1][j_2]}\varepsilon _{i_2j_2}\delta_{i_1}^{(k_1}\delta_{j_1}^{k_2)}-\\
&-(-1)^{[i_1][i_2]}\varepsilon _{i_1j_1}\delta_{i_2}^{(k_1}\delta_{j_2}^{k_2)}+(-1)^{[i_1][i_2]+[j_1][j_2]}\varepsilon _{i_1j_2}\delta_{i_2}^{(k_1}\delta_{j_1}^{k_2)}\;.
\end{aligned}
\end{equation}
A comparison of \eqref{MatBas} and \eqref{AdBas} suggests a convenient relation between the components of the basis elements $M_{ij}$ of $osp(M|N)$ in the adjoint and defining representations:
\begin{equation}\label{MAdFcor}
(M_{i_1i_2})^{k_1k_2}{}_{j_1j_2}=
2(M_{i_1i_2})^{(k_1}{}_{(j_1}\delta^{k_2)}_{j_2)}=4\Sym_{1\leftrightarrow 2}\varepsilon _{i_2j_1}\delta^{k_1}_{i_1}\delta^{k_2}_{j_2}\;,
\end{equation}
where $\Sym_{1\leftrightarrow 2}$ denotes (anti)symmetrisation over the pairs of indices $(i_1,i_2)$, $(j_1,j_2)$ and $(k_1,k_2)$.
 Using \eqref{MAdFcor}, we deduce the following form of the components of $\hCad$:
\begin{align}
(\hC_{\ad})^{k_1k_2k_3k_4}{}_{m_1m_2m_3m_4}&=g^{i_1i_2j_1j_2}(M_{i_1i_2}\otimes M_{j_1j_2})^{k_1k_2k_3k_4}{}_{m_1m_2m_3m_4}\label{CadFbas}\\
&=(-1)^{([j_1]+[j_2])([m_1]+[m_2])}g^{i_1i_2j_1j_2}(M_{i_1i_2})^{k_1k_2}{}_{m_1m_2}(M_{j_1j_2})^{k_3k_4}{}_{m_3m_4}\nonumber\\
&=4(-1)^{([j_1]+[j_2])([m_1]+[m_2])}g^{i_1i_2j_1j_2}(M_{i_1i_2})^{(k_1}{}_{(m_1}\delta^{k_2)}_{m_2)}(M_{j_1j_2})^{(k_3}{}_{(m_3}\delta^{k_4)}_{m_4)}\;,\nonumber
\end{align}
As a result, we find a connection between the components of the split Casimir operator in the adjoint and defining representations:
\begin{equation}\label{Crel}
(\hC_{\ad})^{k_1k_2k_3k_4}{}_{m_1m_2m_3m_4}=
4(\hC_f)^{(k_1k_2)(k_3k_4)}_{13}{}_{(m_1m_2)(m_3m_4)}\; ,
\end{equation}
where
\begin{align}
(\hC_f)^{k_1k_2k_3k_4}_{13}{}_{m_1m_2m_3m_4}&=g^{i_1i_2j_1j_2}
(M_{i_1i_2}\otimes I\otimes M_{j_1j_2}\otimes I)^{k_1k_2k_3k_4}{}_{m_1m_2m_3m_4}\\
&=(-1)^{([j_1]+[j_2])([m_1]+[m_2])}g^{i_1i_2j_1j_2}
(M_{i_1i_2})^{k_1}{}_{m_1}\delta^{k_2}_{m_2}
(M_{j_1j_2})^{k_3}{}_{m_3}\delta^{k_4}_{m_4}\nonumber\;.
\end{align}
 The lower index ''$13$ '' of $\hC_f$ is defined in accordance with \eqref{A13}.

In what follows, we need the operators $\mcP_{\alpha\beta}$ and $\mcK_{\alpha\beta}:V_{(M|N)}^{\otimes 4}\to V_{(M|N)}^{\otimes 4}$
 ($\alpha,\beta=1,\dots, 4$, $\alpha\neq \beta$), which are built from \eqref{Pdef} and \eqref{Kdef}
by \eqref{A123}.
 Direct calculations show that $\mcP_{\alpha,\alpha+1}$ and $\mcK_{\alpha,\alpha+1}$ satisfy the relations between the generators
 $\sigma_\alpha$ and $\kappa_\alpha$ of the Brauer algebra ${\cal B}r_4(\omega)$ :
\begin{equation}\label{Brauer}
\begin{gathered}
\sigma_\alpha^2=I\;,\  \kappa_\alpha^2=\omega \kappa_\alpha\;,\  \sigma_\alpha \kappa_\alpha=\kappa_\alpha \sigma_\alpha=\kappa_\alpha\;, \  \sigma_\alpha\kappa_\alpha=\kappa_\alpha\sigma_\alpha=\kappa_\alpha\;,\quad \alpha=1,\dots,3\;,\\
\sigma_{\alpha}\sigma_{\beta}=\sigma_{\beta}\sigma_{\alpha}\;,\quad \kappa_{\alpha}\kappa_\beta=\kappa_{\beta}\kappa_{\alpha}\;,\quad \sigma_\alpha \kappa_\beta=\kappa_\beta \sigma_\alpha\;,\quad |\alpha-\beta|>1\;,\\
\sigma_\alpha \sigma_{\alpha+1}\sigma_\alpha=\sigma_{\alpha+1}\sigma_{\alpha}
\sigma_{\alpha+1}\,\quad \kappa_\alpha \kappa_{\alpha+1} \kappa_{\alpha}=\kappa_\alpha\;,\quad \kappa_{\alpha+1}\kappa_{\alpha}\kappa_{\alpha+1}=\kappa_{\alpha+1}\;,\\
\sigma_{\alpha}\kappa_{\alpha+1}\kappa_{\alpha}=\sigma_{\alpha+1}
\kappa_{\alpha}\;,\quad \kappa_{\alpha+1}\kappa_\alpha \sigma_{\alpha+1}=\kappa_{\alpha+1}\sigma_\alpha,\quad \alpha=1,\dots,3 \; .
\end{gathered}
\end{equation}
Thus, $\mcP_{\alpha,\alpha+1} = \tau(\sigma_\alpha)$ and $\mcK_{\alpha,\alpha+1} = \tau(\kappa_\alpha)$, where
$\tau$ is a representation of ${\cal B}r_4(\omega)$ in the space $V_{(M|N)}^{\otimes 4}$.
Recall that, by convention \eqref{Aa1a2a3S}, the operators $\mcP_{\alpha\beta}$ for $\beta > \alpha+1$ can be obtained from $\mcP_{\alpha'\beta'}$
 by a consequtive action of adjacent transpositions
 $\mcP_{\gamma,\gamma+1}$. For instance,
\begin{equation}
\mcP_{14}=\mcP_{34}\mcP_{13}\mcP_{34}=
\mcP_{34}\mcP_{23}\mcP_{12}\mcP_{23}\mcP_{34}=
\tau(\sigma_3\sigma_2\sigma_1\sigma_2\sigma_3) \; .
\end{equation}
Besides, being even (or, alternatively by \eqref{Brauer}), the operators $\mcP_{\alpha_1\alpha_2}$ and $\mcP_{\beta_1\beta_2}$ commute for $\alpha_1\neq \beta_1,\beta_2$ and $\alpha_2\neq \beta_1,\beta_2$, and so do $\mcK_{\alpha_1\alpha_2}$ and $\mcK_{\beta_1\beta_2}$
\begin{equation}
\mcP_{\alpha_1\alpha_2}\mcP_{\beta_1\beta_2}=
\mcP_{\beta_1\beta_2}\mcP_{\alpha_1\alpha_2},\quad \mcK_{\alpha_1\alpha_2}\mcK_{\beta_1\beta_2}=
\mcK_{\beta_1\beta_2}\mcK_{\alpha_1\alpha_2}.
\end{equation}

Define the antisymmetrizer
$\mcP^-:V_{(M|N)}^{\otimes 2}\to V_{(M|N)}^{\otimes 2}$
by
\begin{equation}\label{Symdef}
\mcP^-=\frac{1}{2}(\bii-\mcP)\;,
\end{equation}
where $\bii$ and $\mcP$ are given in \eqref{Idef} and \eqref{Kdef}. Then, by \eqref{hCf} and \eqref{Crel},
\begin{equation}\label{hCexp}
\hC_{\ad}=4\mcP_{12}^-\mcP^-_{34}(\hC_f)_{13}\mcP^-_{12}\mcP^-_{34}=\frac{2}{\omega-2}\mcP^-_{12}\mcP^-_{34}(\mcP_{13}-\mcK_{13})\mcP_{12}^-\mcP_{34}^-\;.
\end{equation}

Define the space $V_{\ad}$ of the adjoint representation of the $osp(M|N)$ Lie superalgebra by $V_{\ad}=\mcP^- V_{(M|N)}^{\otimes 2}$. The algebra $osp(M|N)$ coincides with $V_{\ad}$ as a vector space. Now introduce the following operators that act on $V_{\ad}^{\otimes 2}\subset V_{(M|N)}^{\otimes 4}$:
\begin{equation}\label{IPKosp}
\bI=\mcP^-_{12}\mcP^-_{34}\equiv \mcP^-_{12,34}\;,\qquad
\bP=\mcP^-_{12,34}\mcP_{13}\mcP_{24}\mcP^-_{12,34}\;,\qquad
\bK=\mcP_{12,34}^-\mcK_{13}\mcK_{24}\mcP^-_{12,34}\;,
\end{equation}
where we denoted
$
\mcP_{12,34}^-=\mcP_{12}^-\mcP_{34}^-.
$
Note that by (\ref{A123}), the operators (\ref{IPKosp}) (in a way similar to $\mcP$ and $\mcK$) define a Brauer algebra ${\cal B}r_s(\omega)$ representation in the space $V_{\ad}^{\otimes s}$.

The following relations hold for $\bI$, $\bP$ and $\bK$ introduced in \eqref{IPKosp}:
\begin{equation}\label{IPrel}
\bI=\bI\mcP_{12}\mcP_{34}=\mcP_{12}\mcP_{34}\bI,\qquad \bP=\mcP_{12}\mcP_{34}\bP=\bP \mcP_{12}\mcP_{34},
\end{equation}
\begin{equation}\label{PKrel}
\bP^2=\bI,\qquad \bK\bP=\bP\bK=\bK,\qquad \bK^2=\frac{\omega(\omega-1)}{2}\bK,
\end{equation}
\begin{equation}\label{CKrel}
\hC_{\ad}\bP=\bP\hC_{\ad},\qquad \hC_{\ad}\bK=\bK\hC_{\ad}=-\bK.
\end{equation}
The operators \eqref{IPKosp} are invariant with respect to the $osp(M|N)$ Lie superalgebra in the adjoint representation (the definition of ad-invariance is given in (\ref{CasAdInv})). Comparing the last formula in  (\ref{PKrel}) and (\ref{KKK}), we get
\be
\lb{dimOSP}
{\rm sdim} \mathfrak{g} = \frac{\omega(\omega -1)}{2} \; .
\ee

To find the characteristic identity for $\hC_{\ad}$, it is convenient to introduce the symmetric $\hC_+$ and antisymmetric $\hC_-$ projections of $\hC_{\ad}$:
\begin{equation}\label{hCpmosp}
\hC_{\pm}=\frac{1}{2}(\bI\pm\bP)\hC_{\ad},
\end{equation}
which satisfy:
\begin{equation}\label{hCpmPK}
\hC_\pm\hC_\mp=0,\qquad \bP\hC_\pm=\pm\hC_\pm,\qquad \bK\hC_-=\hC_-\bK=0,\qquad \bK \hC_+=\hC_+ \bK=-\bK \; .
\end{equation}
Note that formulas \eqref{PKrel}, \eqref{CKrel} and
\eqref{hCpmPK} were derived for all Lie superalgebras with the nondegenerate Cartan-Killing metric in Section {\bf \ref{21}}.
Substitution of \eqref{hCexp} and \eqref{IPKosp} into \eqref{hCpmosp} gives explicit formulas for the antisymmetric and symmetric parts of $\hC_{\ad}$:
\begin{align}
\hCm&=\frac{1}{\omega-2}\mcP^-_{12,34}(\mcK_{13}\mcP_{24}
-\mcK_{13})\mcP^-_{12,34}\;,\label{hCmExp}\\
\hCp&=\frac{1}{\omega-2}\mcP^-_{12,34}(2\mcP_{24}
-\mcK_{13}-\mcK_{13}\mcP_{24})\mcP^-_{12,34}\;.\label{hCpExp}
\end{align}
\begin{proposition}\label{ospProp}
The antisymmetric $\hCm$ and symmetric $\hCp$ parts of the split Casimir operator of the $osp(M|N)$ Lie superalgebra for $M-N\equiv \omega\neq 0,1,2,4,8$ satisfies:
\begin{equation}\label{hCmChar}
\hC_-^2=-\frac{1}{2}\hC_- \iff \hC_-(\hC_-+\frac{1}{2})=0,
\end{equation}
\begin{equation}\label{Cp3}
\hCp^3=-\frac{1}{2}\hCp^2-\frac{\omega-8}{2(\omega-2)^2}\hCp+\frac{\omega-4}{2(\omega-2)^3}(\bI+\bP-2\bK)\;,
\end{equation}
 \begin{equation}\label{Cp4}
 \hCp^4+\frac{3}{2}\hCp^3+\frac{(\omega+1)(\omega-4)}{2(\omega-2)^2}\hCp^2
 +\frac{\omega^2-12\omega+24}{2(\omega-2)^3}\hCp-
 \frac{(\omega-4)}{2(\omega-2)^3}(\bI+\bP)=0,
 \end{equation}
\begin{equation}\label{Cp5F}
\hCp\Big(\hCp+\bI\Big)\Big(\hCp-\frac{\bI}{\omega-2}\Big)\Big(\hCp+\frac{2\bI}{\omega-2}\Big)\Big(\hCp+\frac{(\omega-4)\bI}{2(\omega-2)}\Big)=0.
\end{equation}
The split Casimir operator $\hCad=\hCm+\hCp$ for $\omega\neq 0,1,2,4,6,8$ satisfies:
\begin{equation}\label{CPolyF}
\hC_{\ad}(\hC_{\ad}+\frac{1}{2})(\hC_{\ad}+1)(\hC_{\ad}-\frac{1}{\omega-2})(\hC_{\ad}+\frac{2}{\omega-2})(\hC_{\ad}+\frac{\omega-4}{2(\omega-2)})=0\;.
\end{equation}
\end{proposition}
\begin{proof}

For $M-N\equiv\omega=2$ the Cartan-Killing metric \eqref{Kil} of $osp(M|N)$ is degenerate, so this case is excluded from consideration. The special cases of $\omega=0,1,4,6,8$ are considered later.

Identity \eqref{hCmChar} for $osp(M|N)$ is a special case of \eqref{Cm2byCm}, which holds for all Lie superalgebras with the nondegenerate Cartan-Killing metric. Note also a useful consequence of \eqref{hCmChar}:
\begin{equation}\label{Cmk}
\hC_-^k=\big(-\frac{1}{2}\big)^{k-1}\hC_-,\quad k\ge 1.
\end{equation}

Using the explicit formula \eqref{hCpExp} for $\hC_+$, one can directly calculate an expression for $\hC_+^2$:
\begin{equation}\label{Cp2}
\hCp^2=\frac{1}{(\omega-2)^2}(\bI+\bP+\bK)-\frac{1}{\omega-2}\hCp+\frac{\omega-8}{2(\omega-2)^2}\mcP_{12,34}^-(\mcK_{13}\mcP_{24}+\mcK_{13})\mcP_{12,34}^-\;.
\end{equation}
If $\omega=8$, then the last term in \eqref{Cp2} is nullified and the characteristic identity for $\hC_+$ takes the form:
\begin{equation}\label{Cp2o8}
\hCp^2=-\frac{1}{6}\hCp+\frac{1}{36}(\bI+\bP+\bK)
\end{equation}
If $\omega\neq 2, 8$, then multiplication of  \eqref{Cp2} by $\hC_+$ yields the third-degree identity \eqref{Cp3}. Note that for $\omega=4$ the last term in \eqref{Cp3} is zero, hence in this case \eqref{Cp3} is the characteristic identity for $\hC_+$ that has the following explicit form:
\begin{equation}\label{Cp3o4}
\hCp^3=-\frac{1}{2}\hCp^2+\frac{1}{2}\hCp\;.
\end{equation}
To obtain a characteristic identity for $\hC_+$ when $\omega\neq 2,4,8$, we get rid of $\bP$ and $\bK$ in \eqref{Cp3}. Multiplying \eqref{Cp3} by $\hCp$ and using \eqref{hCpmPK}, we can express $\bK$ in terms of $\hCp$:
\begin{equation}\label{bK}
\frac{\omega -4}{(\omega-2)^3}\bK=\hCp^4+\frac{1}{2}\hCp^3+\frac{\omega-8}{2(\omega-2)^2}
\hCp^2-\frac{\omega-4}{(\omega-2)^3}\hCp\;.
\end{equation}
 Substitution of $\bK$ from \eqref{bK} into \eqref{Cp3} gives \eqref{Cp4}. Multiplying both sides of \eqref{bK} by $(\hCp+\bI)$ and using the last relation in \eqref{hCpmPK}, we get the characteristic identity for $\hCp$:
\begin{equation}\label{Cp5}
\hCp^5+\frac{3}{2}\hCp^4+
\frac{(\omega+1)(\omega-4)}{2(\omega-2)^2}\hCp^3+
\frac{\omega^2-12\omega+24}{2(\omega-2)^3}\hCp^2-
\frac{\omega-4}{(\omega-2)^3}\hCp=0,
\end{equation}
which can be rewritten in the following form:
\begin{equation}\label{Cp5E}
\hCp^5=-\frac{3}{2}\hCp^4-\frac{(\omega+1)(\omega-4)}{2(\omega-2)^2}\hCp^3-\frac{\omega^2-12\omega+24}{2(\omega-2)^3}\hCp^2+\frac{\omega-4}{(\omega-2)^3}\hCp.
\end{equation}
Now, \eqref{Cp5F} is the result of factorizing \eqref{Cp5}. For further calculations, we also need an expression for $\hCp^6$, which can be derjved by multiplying \eqref{Cp5E} by $\hCp$ and using the known polynomial for $\hCp^5$ from \eqref{Cp5E}:
\begin{equation}\label{Cp6E}
\begin{aligned}
\hCp^6&=\frac{(7\omega^2-30\omega+44)}{4(\omega-2)^2}\hCp^4+\frac{3\omega^3-17\omega^2+30\omega-24}{4(\omega-2)^3}\hCp^3\\
&+\frac{3\omega^2-32\omega+56}{4(\omega-2)^3}\hCp^2-\frac{3(\omega-4)}{2(\omega-2)^3}\hCp.
\end{aligned}
\end{equation}

Our next goal is to find a characteristic polynomial for the split Casimir operator $\hC_{\ad}=\hC_++\hC_-$ by using the expressions obtained. We look for such an expression in the form of a polynomial in $\hC_{\ad}$ of degree six with arbitrary coefficients $\alpha_i$:
\begin{equation}\label{ArbPoly}
\hC_{\ad}^6+\alpha_5\hC_{\ad}^5+\alpha_4\hC_{\ad}^4
+\alpha_3\hC_{\ad}^3+\alpha_2\hC_{\ad}^2+\alpha_1\hC_{\ad}+\alpha_0\;.
\end{equation}
We need to find $\alpha_i$ that nullify \eqref{ArbPoly}. The first formula in
\eqref{hCpmPK} implies $\hC^k=\hC_+^k+\hC_-^k$,
thus equating the polynomial \eqref{ArbPoly} to zero yields the equation
\begin{align*}
&\hC_{+}^6+\alpha_5\hC_{+}^5+\alpha_4\hC_{+}^4+\alpha_3\hC_{+}^3
+\alpha_2\hC_{+}^2+\alpha_1\hC_{+}\\
+&\hC_{-}^6+\alpha_5\hC_{-}^5+\alpha_4\hC_{-}^4
+\alpha_3\hC_{-}^3+\alpha_2\hC_{-}^2+\alpha_1\hC_{-}+\alpha_0=0\;.
\end{align*}
Using the expessions for $\hC_{+}^{5,6}$ in terms of $\hC_{+}^{4,3,2,1}$ given by \eqref{Cp5E} and \eqref{Cp6E} as well as the expressions for $\hC_{-}^{6,5,4,3,2}$ in terms of $\hC_{-}$ from \eqref{Cmk} and setting the coefficients of those operators to zero, we get the values of $\alpha_i$:
\begin{align*}
\alpha_0&=0\;, & \alpha_1&=-\frac{\omega-4}{2(\omega-2)^3}\;, & \alpha_2&=\frac{\omega^2-16\omega+40}{4(\omega-2)^3}\;,\\
\alpha_3&=\frac{\omega^3-3\omega^2-22\omega+56}{4(\omega-2)^3}\;, &
\alpha_4&=\frac{5\omega^2-18\omega+4}{4(\omega-2)^2}\;, & \alpha_5&=2\;.
\end{align*}
Thus, the characteristic identity for $\hC_{\ad}$ takes the form:
\begin{equation}\label{CPoly}
\begin{gathered}
\hC_{\ad}^6+2\hC_{\ad}^5+\frac{5\omega^2-18\omega+4}{4(\omega-2)^2}\hC_{\ad}^4+\frac{\omega^3-3\omega^2-22\omega+56}{4(\omega-2)^3}\hC_{\ad}^3+\\
+\frac{\omega^2-16\omega+40}{4(\omega-2)^3}\hC_{\ad}^2-\frac{\omega-4}{2(\omega-2)^3}\hC_{\ad}=0\;.
\end{gathered}
\end{equation}
The roots of the polynomial on the left of \eqref{CPoly} may be found explicitly:
\begin{equation}\label{Croots}
a_1=0,\quad a_2=-\frac{1}{2},\quad a_3=-1,\quad a_4=\frac{1}{\omega-2},\quad a_5=\frac{-2}{\omega-2},\quad a_6=\frac{4-\omega}{2(\omega-2)}\;.
\end{equation}
Note that degenerate roots appear for the following values of $\omega$:
\begin{equation}\label{ospDegRoots}
\begin{aligned}
\omega=0&\implies a_2=a_4=-\frac{1}{2},\quad a_3=a_6=-1, &\qquad \omega=1&\implies a_3=a_4=-1,\\
\omega=4&\implies a_1=a_6=0,\quad a_3=a_5=-1, &
\omega=6&\implies a_2=a_5=-\frac{1}{2},\\
\omega=8&\implies a_5=a_6=-\frac{1}{3}.
\end{aligned}
\end{equation}
Therefore, the cases $\omega=0,1,4,6,8$
are to be considered separately (see Remark below).

Using the roots \eqref{Croots} of the characteristic polynomial in the left-hand side of \eqref{CPoly}, we can rewrite the characteristic identity \eqref{CPoly} for $\omega\neq 0,1,4,6,8$ as \eqref{CPolyF}.
\end{proof}
\noindent
{\bf Remark.} In order to get characteristic identities for $\hC_{\ad}$ when $\omega=0,1,4,6,8$, we need for all the degenerate roots in \eqref{ospDegRoots} to leave in \eqref{CPolyF} only one of the parentheses corresponding to such a root. This operation turns the left hand side of \eqref{CPolyF} into the following polynomials in $\hC_{\ad}$:
\begin{equation}\label{CPolyFo08}
\begin{aligned}
\omega&=0: & &\hC_{\ad}(\hC_{\ad}+\frac{1}{2})(\hC_{\ad}+1)(\hC_{\ad}-1)=\frac{1}{2}\bK\neq 0,\\
\omega&=1: & &\hC_{\ad}(\hC_{\ad}+\frac{1}{2})(\hC_{\ad}+1)(\hC_{\ad}-2)(\hC_{\ad}+\frac{3}{2})=-\frac{3}{2}\bK\neq 0,\\
\omega&=4: & &\hC_{\ad}(\hC_{\ad}+\frac{1}{2})(\hC_{\ad}+1)(\hC_{\ad}-\frac{1}{2})=0,\\
\omega&=6: & &\hC_{\ad}(\hC_{\ad}+\frac{1}{2})(\hC_{\ad}+1)(\hC_{\ad}-\frac{1}{4})(\hC_{\ad}+\frac{1}{4})=0,\\
\omega&=8: & &\hC_{\ad}(\hC_{\ad}+\frac{1}{2})(\hC_{\ad}+1)(\hC_{\ad}-\frac{1}{6})(\hC_{\ad}+\frac{1}{3})=0.
\end{aligned}
\end{equation}
The equalities in the right-hand side of \eqref{CPolyFo08} are obtained by substitution $\hC_{\ad}=\hCp+\hCm$ and using \eqref{Cmk}.
From \eqref{CPolyFo08}, we see that for $\omega=0,1$ the characteristic polynomial for $\hC_{\ad}$ of the form \eqref{charct} does not exist; hence $\hCad$ is not diagonalizable. Correspondingly, the representation $\ad^{\otimes 2}$ of $osp(M|N)$ in this case is not completely reducible. It can also be deduced from the fact that for $\omega=0,1$ the ad-invariant operator $\bK$ is nilpotent (as $\bK^2=0$) and, therefore, not diagonalizable.

The form \eqref{CPolyF} of the characteristic identity of $\hCad$ for $\omega\neq 0,1,4,6,8$ allows us to construct projectors onto invariant subspaces of $V_{\ad}\otimes V_{\ad}$ by \eqref{GenPr}, where $p=6$ and $a_i$ are the roots \eqref{Croots} of the characteristic equation \eqref{CPolyF}. Using \eqref{hCpmPK},\eqref{Cp3}, \eqref{Cmk}, 
and \eqref{Cp5E}, we find explicit expressions for the projectors \eqref{GenPr} in terms of $\bI$, $\bP$, $\bK$, $\hC_+$, $\hC_-$:
\begin{equation}\label{proj}
\begin{aligned}
\proj_1&=\frac{1}{2}(\bI-\bP)+2\hC_{-}\;, \qquad\quad \proj_2=-2\hC_{-}\;, \qquad\quad \proj_3=\frac{2\bK}{(\omega-1)\omega} \;, \\
\proj_4&=\frac{2}{3}(\omega-2)\hC_{+}^2+\frac{\omega}{3}\hC_{+}
+\frac{(\omega-4)(\bI+\bP)}{3(\omega-2)}-
\frac{2(\omega-4)\bK}{3(\omega-2)(\omega-1)} \;, \\
\proj_5&=-\frac{2(\omega-2)^2}{3(\omega-8)}\hC_{+}^2-
\frac{(\omega-2)(\omega-6)}{3(\omega-8)}\hC_{+}
+\frac{(\omega-4)(\bI+\bP)}{6(\omega-8)}+\frac{2\bK}{3(\omega-8)}\;,\\
\proj_6&=\frac{4(\omega-2)}{\omega-8}\hC_{+}^2+\frac{4}{\omega-8}\hC_{+}
-\frac{4(\bI+\bP)}{(\omega-2)(\omega-8)}-\frac{8(\omega-4)\bK}{\omega(\omega-2)(\omega-8)}\;.
\end{aligned}
\end{equation}
The images of $\proj_1$ and $\proj_2$ are contained in the antisymmetric part $\bP_-(\Vad^{\otimes 2})$, while the images of $\proj_i$, $(i=3,...,6)$, lie within the symmetric part $\bP_+(\Vad^{\otimes 2})$ of $\Vad^{\otimes 2}$ where $\bP_{\pm}=\frac{1}{2}(\bI\pm\bP)$. Note that for $\omega=4,6$ all the projectors \eqref{proj} are well defined and constructed from the ad-invariant operators $\bI$, $\bP$, $\bK$, $\hCm$, $\hCp$ and $\hCp^2$; hence they are projectors onto the invariant subspaces of the representation $\ad^{\otimes 2}$ of $osp(M|N)$. Besides, although $\proj_5$ and $\proj_6$ are not formally defined for $\omega=8$, substitution of the explicit expressions for $\bI$, $\bP$, $\bK$, $\hCp$ and $\hCp^2$ into \eqref{proj} and the subsequent cancellation of the pole at $\omega=8$ yield:
\begin{align}
\proj_5&=\frac{1}{6}(1-(\mcP_{14}+\mcP_{23}+\mcP_{13}+\mcP_{24})+\mcP_{13}\mcP_{24})\mcP_{12,34},\label{ospProj5}\\
\proj_6&=\frac{4}{\omega-2}\mcP_{12,34}\mcK_{13}\big[\frac{1}{2}(1+\mcP_{24})-\frac{1}{\omega}\mcK_{24}\big]\mcP_{12,34}.
\end{align}
It is worth pointing out that expression \eqref{ospProj5} for the projector $\proj_5$ does not depend on $\omega$ and coincides with the full (anti)symmetrizer of $V_{(M|N)}^{\otimes 4}$.

To build projectors onto the generalized eigenspaces of $\hCad$ for $\omega=0,1$, we use the characteristic identity \eqref{charct2}. For $\omega=0$ it takes the form:
\begin{equation*}
\hCad(\hCad+\frac{1}{2})(\hCad+1)^2(\hCad-1)=0,
\end{equation*}
so one needs to put $a_1=1$, $a_2=-\frac{1}{2}$, $a_3=-1$, $a_4=1$, $k_1=k_2=k_4=1$, $k_3=1$ in \eqref{charct2}. Then \eqref{GenPr2} gives the projectors onto the generalized eigenspaces of $\hCad$:
\begin{equation}\label{projospo0}
\begin{aligned}
\proj_1&=\frac{1}{2}(\bI-\bP)+2\hCm, &\quad
\proj_2&=-2\hCm+\frac{2}{3}(\bI+\bP)+\frac{4}{3}\bK-\frac{4}{3}\hCp^2,\\
\proj_3&=-\frac{1}{4}(\bI+\bP)-\frac{5}{4}\bK-\frac{1}{2}\hCp+\hCp^2, &
\proj_4&=\frac{1}{12}(\bI+\bP)-\frac{1}{12}\bK+\frac{1}{2}\hCp+
\frac{1}{3}\hCp^2.
\end{aligned}
\end{equation}
Here the operators $\proj_1$, $\proj_2$ and $\proj_4$ extract eigenspaces of $\hCad$, while $\proj_3$ projects $\Vad^{\otimes 2}$ onto the generalized eigenspace of $\hCad$. We thus conclude that the restriction of the representation $\ad^{\otimes 2}$ to $\proj_3(\Vad^{\otimes 2}$ is reducible but not completely reducible. It is also worth noting that $\proj_2$ given in \eqref{projospo0} is a linear combination of symmetric operators $(\bI+\bP)$, $\bK$, $\hCp$ and an antisymmetric operator $\hCm$. As pointed out above (see Section \ref{21}), the symmetric and antisymmetric parts of $\Vad^{\otimes 2}$ are invariant under the action of any Lie superalgebra $\mfg$ in the representation $\ad^{\otimes 2}$, so $\proj_2$ may be invariantly split into its symmetric and antisymmetric parts:
\begin{equation}\label{proj2ospo1}
\proj_2^{(+)}=\frac{2}{3}(\bI+\bP)+\frac{4}{3}\bK-\frac{4}{3}\hCp, \qquad\qquad\qquad
\proj_2^{(-)}=-2\hCm.
\end{equation}

Finally, for $osp(M|N)$ in the case of $\omega\equiv M-N=0$, the needed projectors are:
\begin{equation*}
\begin{aligned}
\proj_1\equiv&\proj_1^{(-)}=\frac{1}{2}(\bI-\bP)+2\hCm, \qquad &
&\proj_2^{(+)}=\frac{2}{3}(\bI+\bP)+\frac{4}{3}\bK-\frac{4}{3}\hCp^2,\\
&\proj_2^{(-)}=-2\hCm, &
\proj_3\equiv &\proj_3^{(+)}=-\frac{1}{4}(\bI+\bP)-\frac{5}{4}\bK-\frac{1}{2}\hCp+\hCp^2,\\
&\ &
\proj_4\equiv &\proj_4^{(+)}=\frac{1}{12}(\bI+\bP)-\frac{1}{12}\bK+\frac{1}{2}\hCp+\frac{1}{3}\hCp^2.
\end{aligned}
\end{equation*}

If $\omega=1$, then the characteristic identity \eqref{charct2} for $\hCad$ is:
\begin{equation*}
\hCad(\hCad+\frac{1}{2})(\hCad+1)^2(\hCad-2)(\hCad+\frac{3}{2})=0,
\end{equation*}
which implies $a_1=0,\ a_2=-\frac{1}{2},\ a_3=-1,\ a_4=2,\ a_5=-\frac{3}{2}, k_1=k_2=k_4=k_5=1,\ k_3=2$ in \eqref{charct2}. The projectors onto the generalised eigenspaces of $\hCad$ are then specified by \eqref{GenPr2}:
\begin{equation*}
\begin{aligned}
\proj_1&=\frac{1}{2}(\bI-\bP)+2\hCm, \qquad& \proj_3&=\bI+\bP-\frac{10}{3}\bK+\frac{1}{3}\hCp-\frac{2}{3}\hCp^2,\\
\proj_2&=-2\hCm, & \proj_4&=\frac{1}{14}(\bI+\bP)-\frac{2}{21}\bK+\frac{5}{21}\hCp+\frac{2}{21}\hCp^2,\\
&\ &
\proj_5&=-\frac{4}{7}(\bI+\bP)+\frac{24}{7}\bK-\frac{4}{7}\hCp+\frac{4}{7}\hCp^2.
\end{aligned}
\end{equation*}
Here the operators $\proj_1$, $\proj_2$, $\proj_4$ and $\proj_5$ extract eigenspaces of $\hCad$, while $\proj_3$ projects $\Vad^{\otimes 2}$ onto the generalized eigenspace of $\hCad$. It can be then concluded that $\proj_3(\Vad^{\otimes 2})$ is not a space of an irreducible or completely reducible representation of $osp(M|N)$ for $\omega=1$.

In order to find the dimensions of the invariant subspaces, we need to calculate the traces and supertraces of $\proj_1,\dots, \proj_6$. First, we calculate the following auxiliary traces and supertraces (here we also use the notation $\xi\equiv M+N$):
\begin{equation}\label{AuxTr}
\begin{aligned}
\tr \bI &=\frac{1}{4}(\xi^2-\omega)^2\;, & \str \bI&= \frac{1}{4}\omega^2(\omega-1)^2\;,\\
\tr \bP &=\frac{1}{2}\omega(\omega-1)\;, & \str \bP&=\frac{1}{2}\omega(\omega-1)\;,\\
\tr \bK &=\frac{1}{2}\omega(\omega-1)\;, & \str \bK&=\frac{1}{2}\omega(\omega-1)\;,\\
\tr \hCm&=-\frac{1}{4}(\xi^2-\omega)\;, & \str \hCm&=-\frac{1}{4}\omega(\omega-1)\;,\\
\tr \hCp&=\frac{1}{4}(\xi^2-\omega)\;, & \str\hCp&=\frac{1}{4}\omega(\omega-1)\;,\\
\tr \hCp^2&=\frac{2\xi^4+(\omega^2-16\omega +20)\xi^2+\omega^3+4\omega^2-12\omega}{8(\omega-2)^2}\;, & \str\hCp^2&=\frac{3}{8}\omega(\omega-1)\;.
\end{aligned}
\end{equation}
 Here the formulas for $\str$ in the second column correspond to the general formulas \eqref{strGenMany} in view of \eqref{dimOSP}.
From \eqref{AuxTr} and \eqref{proj}, we get the traces
\begin{align}
\tr\proj_1&=\frac{1}{8}(\xi^4-2\xi^2(\omega+2)-\omega(\omega-6))\;,&\tr\proj_4&=\frac{1}{12}(\xi^4-10\xi^2+3\omega(\omega-2))\;,\nonumber\\
\tr\proj_2&=\frac{1}{2}(\xi^2-\omega)\;, & \tr\proj_5&=\frac{1}{24}(\xi^4+2\xi^2(3\omega-4)+3\omega(\omega-2)\;,\nonumber\\
\tr\proj_3&=1\;, & \tr\proj_6&=\frac{1}{2}(\xi^2+\omega-2)\label{osptr}
\end{align}
snd supertraces of $\proj_i$:
\begin{equation}\label{ospstr}
\begin{aligned}
\str\proj_1&=\frac{1}{8}\omega(\omega-1)(\omega+2)(\omega-3)\;,\qquad&\str\proj_4&=\frac{1}{12}\omega(\omega+1)(\omega+2)(\omega-3)\;,\\
\str\proj_2&=\frac{1}{2}\omega(\omega-1)\;,&\str\proj_5&=\frac{1}{24}\omega(\omega-1)(\omega-2)(\omega-3)\;,\\
\str\proj_3&=1\;,&\str\proj_6&=\frac{1}{2}(\omega-1)(\omega+2)\\
\end{aligned}
\end{equation}
for $\omega\neq 0,1$. The dimension $\dim_{\oo}V_i$ of the even part of the invariant subspace $V_i=V_{i\oo}\oplus V_{i\ol}\subseteq V_{(M|N)}^{\otimes 4}$ extracted by $\proj_i$ is $\dim_{\oo}V_i=\frac{1}{2}(\tr \proj_i+\str \proj_i)$, while the dimension of the odd part $V_{i\ol}$ of $V_i$ is $\dim_{\ol}V_i=\frac{1}{2}(\tr\proj_i-\str\proj_i)$. Using \eqref{osptr} and \eqref{ospstr}, and substituting $\omega=M-N$ and $\xi=M+N$, we obtain the following values for dimensions of the invariant subspaces:
\begin{equation}\label{ospDims0}
\begin{array}{c}
\dim_{\oo}V_1= \frac{1}{8}M(M-1)(M+2)(M-3)+\frac{1}{8}N(N+1)(N-2)(N+3) + \\
  +\frac{1}{4}MN(3MN+M-N+1)\;,\\ [0.2cm]
\dim_{\oo}V_2=  \frac{1}{2}M(M-1)+\frac{1}{2}N(N+1)\;, \;\;\;\;\;\;\;\;
\dim_{\oo}V_3=1\;,\\ [0.2cm]
\dim_{\oo}V_4=  \frac{1}{12}M(M+1)(M+2)(M-3)+\frac{1}{12}N(N-1)(N-2)(N+3)+\\
+  \frac{1}{2}MN(MN-1)\;,\\ [0.2cm]
\dim_{\oo}V_5=  \frac{1}{24}M(M-1)(M-2)(M-3)+\frac{1}{24}N(N+1)(N+2)(N+3) + \\
+\frac{1}{4}MN(M-1)(N+1)\;,\\ [0.2cm]
\dim_{\oo}V_6=  \frac{1}{2}(M-1)(M+2)+\frac{1}{2}N(N-1)\;,\\
\end{array}
\end{equation}
\begin{equation}
\label{ospDims1}
 \begin{array}{c}
\dim_{\ol}V_1 =\frac{1}{2}MN\bigl(M(M-1)+(N-1)(N+2)\bigr)\;, \;\;\;\;\;\;\;
 \dim_{\ol}V_2 =MN\;, \\ [0.2cm]
\dim_{\ol}V_3=0\;, \;\;\;\;\;\;\;  \dim_{\ol}V_4 =\frac{1}{3}MN(M^2+N^2-5)\;,
\\ [0.2cm]
\dim_{\ol}V_5 =\frac{1}{6}MN\big((M-1)(M-2)+(N+1)(N+2)\big)\;, \;\;\;\;\;\;\;  \dim_{\ol}V_6 =MN\;.
\end{array}
\end{equation}
Note that the substitution $\omega=M$ (which implies $N=0$) into the identites \eqref{hCmChar}--\eqref{CPolyF} yields analogous identities for the $so(M)$ Lie algebra that are given in \cite{IsKr} and \cite{IsPr}, and the dimensions \eqref{ospDims0} of the invariant subspaces of the $osp(M|N)$ Lie superalgebra representation $\ad^{\otimes 2}$ transform into the corresponding dimensions of the invariant subspaces of $so(M)$. Analogously, the substitution $\omega=-N$ (which means $M=0$) transforms identities \eqref{hCmChar}--\eqref{CPolyF} into analogous identities for the algebra $sp(N)$ while \eqref{ospDims0} gives the dimensions of the invariant subspaces of the $sp(N)$ Lie algebra representation $\ad^{\otimes 2}$. Indeed, the substitutions $M=0$ and $N=0$ nullify the dimensions of the odd parts of the invariant subspaces \eqref{ospDims1}, which corresponds to the transition from the $osp(M|N)$ Lie superalgebra to the $so(M)$ or $sp(N)$ Lie algebras.

\section{The $s\ell(M|N)$ Lie superalgebra}\label{Sec4}
The $s\ell(M|N)$ Lie superalgebra (where $M\neq N$) is defined as the algebra of the operators $A:V_{(M|N)}\to V_{(M|N)}$ that satisfy:
\begin{equation}\label{str0mat}
\str A=0,
\end{equation}
and the Lie superbracket of which is given by \eqref{AsToLie}. It is known that $s\ell(N,N)$ is not simple, and $s\ell(M|N)\cong s\ell(N|M)$, so we restrict ourselves to $\omega=M-N>0$.

To build a basis of $s\ell(M|N)$, we use the same method as in the $osp(M|N)$ case, that is, we build a projector onto the space of solutions of \eqref{str0mat}. Consider the operators $\proj_0$ and $\proj_I$ that act on $\End(V_{(M|N)})$ by
\begin{equation}
\begin{aligned}
\proj_0 (E)&=E-\frac{\str E}{M-N}I,\\
\proj_I (E)&=\frac{\str E}{M-N}I,
\end{aligned}
\end{equation}
for any $E\in \End(V_{(M|N)})$ and $I$ being the identity operator acting in $V_{(M|N)}$. Clearly,
\begin{equation}
\begin{aligned}
\str \proj_0(E)&=0,\\
\str \proj_I(E)&=\str E,
\end{aligned}
\end{equation}
so $\proj_A \, \proj_B= \delta_{AB} \proj_A$ ($A,B=0,I$). Therefore, $\proj_0$ and $\proj_I$ comprise a full system of projectors in $\End(V_{(M|N)})$. Analogously to the $osp(M|N)$ case, \eqref{str0mat} can be rewritten in terms of $\proj_A$ as $\proj_I(A)=0$ or $\proj_0(A)=A$, which implies the following general solution of \eqref{str0mat}:
\begin{equation}\label{slGenMat}
A^a{}_c=E^a{}_c-\frac{(-1)^{[b]}\, E^b{}_b}{M-N} \, \delta^a_c,
\end{equation}
where $E$ is an arbitrary element of $\End(V_{(M|N)})$.

 Consider the matrix identities $e_{ij}:V_{(M|N)}\to V_{(M|N)}$ with the components
 \begin{equation}\label{MatUn1}
 (e_{ij})^a{}_b=\delta^a_i\delta_{jb} \; .
 \end{equation}
 Note that unlike the matrix identites $e_i{}^j$ of the $osp(M|N)$ case, both indices of $e_{ij}$ are lower. The matrices of the basis elements $(T_{ij})^a{}_b$
 of $s\ell(M|N)$ in the defining representation are obtained by substituting $E^a{}_b \to (e_{ij})^a{}_b$
 into \eqref{slGenMat}:
\begin{equation}\label{slBas}
(T_{ij})^a{}_b=(e_{ij})^a{}_b-
\frac{(-1)^{[i][j]}}{M-N}\delta_{ij}\delta^a_b.
\end{equation}
The degree of $T_{ij}$ coincides with that of $e_{ij}$ and equals $[i]+[j]$. Besides, we have the following equality for $T_{ij}$:
\begin{equation}\label{slBasCon}
{\rm Tr}(T) \equiv T_{ii}=0.
\end{equation}
In order for the vectors $X=X^{ij}T_{ij}$ of the algebra $s\ell(M|N)$ to correspond uniquely to their coordinates $X^{ij}$, we require the numbers $X^{ij}$ to satisfy the condition
\begin{equation}\label{slComCon}
(-1)^{[i]}X^{ii}=0,
\end{equation}
which has the following advantage: for any $X = X^{ij}T_{ij} \in s\ell(M|N)$, in the defining representation $X=
 X^{ij}T_{ij}=X^{ij}e_{ij} \in s\ell(M|N).$
If the vector $X\in \Mat_{M+N}(\mathbb{C})$ is expanded over the basis $\{e_{ij}\}_{i,j=1}^{M+N}$, \eqref{slComCon} means that $X$ lies in $s\ell(M|N)$. We also impose the conditions analogous to \eqref{slBasCon} on the coordinates $Y_{ij}$ of the covectors $Y$ in the dual basis of \eqref{slBas}:
\begin{equation}\label{slCovCon}
{\rm Tr}(Y)= Y_{ii}=0.
\end{equation}

Calculating the Lie superbracket \eqref{AsToLie} of $T_{ij} \in s\ell(M|N)$ defined in \eqref{slBas}, we obtain:
\begin{equation}
\lbrack T_{i_1i_2},T_{j_1j_2}\rbrack =
\delta_{j_1 i_2}T_{i_1 j_2}-
(-1)^{([i_1]+[i_2])([j_1]+[j_2])}\delta_{i_1 j_2}T_{j_1 i_2}
=T_{k_1k_2} \, X^{k_1k_2}{}_{i_1i_2,j_1j_2} \; ,
\end{equation}
where the structure constants $X^{k_1k_2}{}_{i_1i_2,j_1j_2}$
are written explicitly as
\begin{equation}\label{slStruc}
X^{k_1k_2}{}_{i_1i_2,j_1j_2}=\delta_{j_1i_2}\delta_{i_1}^{k_1}\delta_{j_2}^{k_2}-(-1)^{([i_1]+[i_2])([j_1]+[j_2])}\delta_{i_1j_2}\delta_{j_1}^{k_1}\delta_{i_2}^{k_2}.
\end{equation}
One can check that the pairs of indices $(i_1,i_2)$, $(j_1,j_2)$ and $(k_1,k_2)$ in \eqref{slStruc}
satisfy \eqref{slBasCon},
\eqref{slCovCon} and \eqref{slComCon}, respectively.

The Cartan-Killing metric \eqref{KilDefG} of $s\ell(M|N)$ in the basis \eqref{slBas} is calculated by \eqref{KilDefG}:
\begin{equation}\label{slKil}
{\sf g}_{i_1i_2,j_1j_2}=2\omega \big((-1)^{[i_1][j_2]}\delta_{j_1i_2}\delta_{i_1j_2}-
\frac{(-1)^{[i_1]+[j_2]}}{\omega}\delta_{i_1i_2}\delta_{j_1j_2}\big),
\end{equation}
where $\omega=M-N$, as in the $osp(M|N)$ case. For $\omega=0$, the metric \eqref{slKil} is degenerate: $\kg_{i_1i_2,j_1j_2}=-2(-1)^{[i_1]+[j_2]}\delta_{i_1i_2}\delta_{j_1j_2}$. However, as we have mentioned, this case is omitted in our paper. Note also that for $(i_1,i_2)$ and $(j_1,j_2)$ in (\ref{slKil}) the condition (\ref{slCovCon}) holds.

 Let us introduce the projector $\oI$ that acts on $V_{(M|N)}^{\otimes 2}$ with the components
\begin{equation}
\label{barI}
\oI^{i_1i_2}{}_{j_1j_2}=\delta^{i_1}_{j_1}\delta^{i_2}_{j_2}-
\frac{(-1)^{[j_1][j_2]}}{\omega}\delta^{i_1i_2}\delta_{j_1j_2}
\;\;\; \Rightarrow \;\;\; \oI^2 = \oI \; .
\end{equation}
It maps an arbitrary $Y \in V_{(M|N)}^{\otimes 2}$ with the components $Y^{ik}$
to $X\in V_{(M|N)}^{\otimes 2}$ with the components $X^{ik} = \oI^{ik}{}_{j\ell}Y^{j\ell}$
 that satisfy (\ref{slComCon}).
 If we identify the space $s\ell(M|N)$ with $\oI (V_{(M|N)}^{\otimes 2}) \subset V_{(M|N)}^{\otimes 2}$,
 then $\oI$ can be viewed as the identity operator $\oI$ acting in the algebra
 $s\ell(M|N)$.

The components of the inverse Cartan-Killing metric in the basis \eqref{slBas} are defined by \eqref{InvKilG}:
\begin{equation}
{\sf g}^{i_1i_2,j_1j_2}{\sf g}_{j_1j_2,k_1k_2}=\oI^{i_1i_2}{}_{k_1k_2}, \qquad {\sf g}_{i_1i_2,j_1j_2}{\sf g}^{j_1j_2,k_1k_2}=\oI^{k_1k_2}{}_{i_1i_2}.
\end{equation}
Direct calculaitons yield
\begin{equation}\label{InvKilSl}
{\sf g}^{i_1i_2,j_1j_2}=\frac{1}{2\omega}\big((-1)^{[j_1][i_2]}
\delta^{j_1i_2}\delta^{i_1j_2}-\frac{1}{\omega}\delta^{i_1i_2}
\delta^{j_1j_2}\big) \; ,
\end{equation}
where the pairs of indices
$(i_1,i_2)$ and $(j_1,j_2)$
 satisfy (\ref{slCovCon}), as expected.

\subsection{Projectors onto invariant subspaces of the tensor product of two defining representations}

Using \eqref{SpCas}, \eqref{ABmat}, \eqref{slBas} and \eqref{InvKilSl},
we can write the matrix of the split Casimir operator of $s\ell(M|N)$ in the defining representation:
\begin{equation}
\label{hCmat}
\hC_f^{k_1k_2}{}_{m_1m_2}=
\frac{1}{2\omega}\big((-1)^{[m_1][m_2]}\delta^{k_1}_{m_2}\delta^{k_2}_{m_1}
-\frac{1}{\omega}\delta^{k_1}_{m_1}\delta^{k_2}_{m_2}\big).
\end{equation}
 Define the identity operator $\bii$ and the superpermutation $\mcP$ acting in $V_{(M|N)}^{\otimes 2}$:
\begin{align}
\bii&=e_{ii}\otimes e_{jj}
& &\implies &
\bii^{k_1k_2}{}_{m_1m_2}&=\delta^{k_1}_{m_1}\delta^{k_2}_{m_2} \;,\label{Idef2}\\
\mcP&=(-1)^{[j]}e_{ij}\otimes e_{ji} & &\implies &
\mcP^{k_1k_2}{}_{m_1m_2}&=
(-1)^{[k_1][k_2]}\delta^{k_1}_{m_2}\delta^{k_2}_{m_2} \;,\label{Pdef2}
\end{align}
where $e_{ij}$ are the matrix identities \eqref{MatUn1}
on $V_{(M|N)}$. The matrices of $\bii$ and $\mcP$ are presented in the right-hand sides of these equalities. By \eqref{Pdef2},
$\mcP^2=\bii$.
A comparison of \eqref{hCmat} and \eqref{Idef2}, \eqref{Pdef2} shows that $\hC_f$ can be written in terms of $\bii$ and $\mcP$ as
\begin{equation}\label{hCfsl}
\hC_f=\frac{1}{2\omega}\big(\mcP-\frac{1}{\omega}\bii\big).
\end{equation}
Using this formula, we obtain the second-degree characteristic identity for $\hC_f$:
\begin{equation}\label{hCfslchar1}
\hC_f^2+\frac{1}{\omega^2}\hC_f-\frac{\omega^2-1}{4\omega^4}\bii=0.
\end{equation}
The left-hand side of this equality can be factorized, which yields
\begin{equation}\label{hCfslchar}
\big(\hC_f-\frac{\omega-1}{2\omega^2}\bii\big)
\big(\hC_f+\frac{\omega+1}{2\omega^2}\bii\big)=0.
\end{equation}
The projectors $\proj_+$ and $\proj_-$ onto the eigenspaces of $\hC_f$ that correspond to the roots $a_+=\frac{\omega-1}{2\omega^2}$ and $a_-=-\frac{\omega+1}{2\omega^2}$ of equation \eqref{hCfslchar} are built by \eqref{GenPr} with $p=2$:
\begin{equation}\label{slFproj}
\proj_\pm=\pm\Big(\omega\hC_f+\frac{1\pm\omega}{2\omega}\Big)=
\frac{1}{2}(\bii\pm\mcP).
\end{equation}
Thus, $\proj_+$ and $\proj_-$ turn out to be the super-symmetrizer and the super-antisymmetriser in $V_{(M|N)}^{\otimes 2}$.

The solution of the graded Yang-Baxter equation \eqref{grYBE} that is invariant under the action of $s\ell(M|N)$ in the defining representation can be written in several equivalent ways:
\begin{equation}
R(u)=\frac{u+\mcP}{1-u}=\frac{1+u}{1-u}\proj_+-\proj_-
\end{equation}
where $u$ is the spectral parameter. Analogously to the $osp(M|N)$ case, this solution can be expressed as a rational function of $\hC_f$:
\begin{equation}
R(u)=\frac{\proj_++u}{\proj_+-u}=\frac{\omega \hC_f+\frac{1+\omega}{2\omega}+u}{\omega \hC_f+\frac{1+\omega}{2\omega}-u}.
\end{equation}
Note that the solution $R(u)$ is unitary: $\mcP R(u)\mcP R(-u)=\bii$. It is defined up to multiplication by an arbitrary function $f(u)$ for which $f(u)f(-u)=1$.

\subsection{Projectors onto invariant subspaces of the tensor product of two adjoint representations}

In order to make the following calculations more concise, we introduce
the operator $\mcK:V_{(M|N)}^{\otimes 2}\to V_{(M|N)}^{\otimes 2}$ the components of which in the homogeneous basis $\{e_{i_1} \otimes e_{i_2} \}$ of this space are
\begin{equation}\label{PKKKK}
\begin{aligned}
\mcK^{i_1i_2}{}_{j_1j_2}&=(-1)^{[j_1][j_2]}
\delta^{i_1i_2}\delta_{j_1j_2} \, .
\end{aligned}
\end{equation}
One can derive the following identities for $\mcP_{ab}$, $\mcK_{ab}:V_{(M|N)}^{\otimes 4}\to V_{(M|N)}^{\otimes 4}$, $(a,b=1,\dots, 4)$:
\begin{equation}\label{slKrel}
\begin{gathered}
\mcK_{ab}\mcK_{ab}=\omega \mcK_{ab},\quad \mcP_{ab}\mcK_{ad}\mcK_{bc}=\mcP_{cd}\mcK_{ad}\mcK_{bc},\quad \mcK_{ad}\mcK_{bc}\mcP_{ab}=\mcK_{ad}\mcK_{bc}\mcP_{cd} \, , \\[3mm]
\mcK_{ab}\mcP_{ab}\mcK_{bc}=\mcK_{ab}\mcP_{ab}\mcP_{ac}=
\mcP_{ac}\mcP_{bc}\mcK_{bc},\quad \mcK_{ab}\mcP_{bc}\mcK_{bc}=\mcK_{ab}\mcP_{ab}\mcP_{ac}=
\mcP_{ac}\mcP_{bc}\mcK_{bc} \; ,
\end{gathered}
\end{equation}
($\mcP_{ab}$ and $\mcK_{ab}$ are defined in \eqref{A123}). It is worth pointing out that $\oI$, defined in \eqref{barI}, can be written as $
\oI=\bii-\frac{1}{\omega}\mcK
$
where $\bii$ is given in \eqref{Idef2}.

The components of the basis elements $T_{i_1i_2} \in s\ell(M|N)$ in the adjoint representation are precisely the structure constants \eqref{slStruc}:
\begin{equation}\label{slAd2Bas}
(T_{i_1i_2})^{k_1k_2}{}_{j_1j_2}=\delta_{j_1i_2}\delta_{i_1}^{k_1}\delta_{j_2}^{k_2}-(-1)^{([i_1]+[i_2])([j_1]+[j_2])}\delta_{i_1j_2}\delta_{j_1}^{k_1}\delta_{i_2}^{k_2}.
\end{equation}
Now using \eqref{SpCas}, \eqref{slAd2Bas} and \eqref{InvKilSl}, we find an explicit form of $\hC_{\ad}$. In terms of the operators \eqref{Pdef2}
and \eqref{PKKKK}, it can be written as follows:
\begin{equation}\label{kazSL}
\hC_{\ad}=\frac{1}{2\omega}\big(\mcP_{13}+\mcP_{24}-\mcK_{32}-\mcK_{14}\big).
\end{equation}
 In what follows, we need three more operators: $\bK,\ \bP^{\ad}:V_{\ad}^{\otimes 2}\to V_{\ad}^{\otimes 2}$ and $\bP:V_{(M|N)}^{\otimes 4}\to V_{(M|N)}^{\otimes 4}$, $\bK$ is defined by
\begin{equation}
\bK^{i_1i_2i_3i_4}{}_{j_1j_2j_3j_4} = \okg^{i_1i_2i_3i_4}{\sf g}_{j_1j_2j_3j_4},
\end{equation}
and expressing it in terms of the operators \eqref{Pdef2},
 \eqref{PKKKK}, we obtain
\begin{equation}
\bK=\mcK_{32}\mcK_{14}-\frac{1}{\omega}\mcP_{24}\mcK_{12}\mcK_{34}-
\frac{1}{\omega}\mcP_{13}\mcK_{32}\mcK_{14}+
\frac{1}{\omega^2}\mcK_{12}\mcK_{34}.
\end{equation}
Furthermore, $\bK^2=(\omega^2-1)\bK$. Recall that $\mcK_{32}=\mcP_{23}\mcK_{23}\mcP_{23}$,
 $\mcK_{14} = \mcP_{34} \mcP_{23} \mcK_{12} \mcP_{23} \mcP_{34}$
 etc.
The operator $\bP$ permutes the first and third, second and fourth factors in $V_{(M|N)}^{\otimes 4}$
and is defined by
\begin{equation}
\bP=\mcP_{13}\mcP_{24}\implies (\bP)^{i_1i_2i_3i_4}{}_{j_1j_2j_3j_4}=
(-1)^{[i_1][i_3]+[i_1][i_4]+[i_2][i_3]+[i_2][i_4]}
\delta^{i_1}_{j_3}\delta^{i_2}_{j_4}\delta^{i_3}_{j_1}\delta^{i_4}_{j_2},
\end{equation}
so we have $\bP^2=I$
where $I$ is the identity operator acting in $V_{(M|N)}^{\otimes 4}$. Finally, the operator
\begin{equation}
\bPad=\oI_{12}\oI_{34}\bP \; ,
\end{equation}
where $\oI$ is given in \eqref{barI},
is the permutation operator in the space
$V_{ad}\otimes V_{\ad}\subset V_{(M|N)}^{\otimes 4}$.

Note that $\bP$ commutes with both $\hC_{\ad}$ and $\bK$:
\begin{equation}\label{slPCadK}
\bP\hC_{\ad}=\hC_{\ad}\bP,\qquad \bP\bK=\bK=\bK\bP.
\end{equation}

Using $\bP$ and $\bPad$ we define the symmetrizer $\bPad_+$ and antisymmetrizer $\bPad_-$ in $V_{\ad}^{\otimes 2}$:
\begin{equation}\label{slPadDif}
\begin{aligned}
\bPad_\pm&=\frac{1}{2}(\bI\pm \bPad)=\frac{1}{2}(I\pm \bP)\oI_{12}\oI_{34}=\frac{1}{2}\oI_{12}\oI_{34}(I\pm \bP)
\quad \Longrightarrow \\
\bPad_-&=\frac{1}{2}(I-\mcP_{13}\mcP_{24})(I-\frac{1}{\omega}(\mcK_{12}+\mcK_{34})),\\
\bPad_+&=\frac{1}{2}(I+\mcP_{13}\mcP_{24})(I-\frac{1}{\omega}(\mcK_{12}+\mcK_{34})+\frac{1}{\omega^2}\mcK_{12}\mcK_{34}),
\end{aligned}
\end{equation}
where $\bI=\oI_{12}\oI_{34}$ is the identity operator acting in $V_{\ad}^{\otimes 2}$ that satisfies
\begin{equation}
\bI\bPad=\bPad=\bPad\bI,\quad \bI\bK=\bK=\bK\bI,\quad \bI\hC_{\ad}=\hC_{\ad}=\hC_{\ad}\bI.
\end{equation}
Define now the symmetric and antisymmetric parts of the split Casimir operator \eqref{kazSL}:
\begin{equation}\label{slhCpmExp}
\begin{aligned}
\hCp&=\bPad_+\hC_{\ad}=\frac{1}{2}(I+\bP)\hC_{\ad}\\
&=\frac{1}{4\omega}\big(2\mcP_{13}+2\mcP_{24}-(I+\mcP_{13}\mcP_{24})\mcK_{32}-(I+\mcP_{13}\mcP_{24})\mcK_{14}\big),\\
\hCm&=\bPad_-\hC_{\ad}=\frac{1}{2}(I-\bP)\hC_{\ad}\\
&=\frac{1}{4\omega}(\mcP_{13}\mcP_{24}-I)(\mcK_{14}+\mcK_{32})=\frac{1}{4\omega}(\mcK_{14}+\mcK_{32})(\mcP_{13}\mcP_{24}-I).
\end{aligned}
\end{equation}
By \eqref{slKrel} and \eqref{kazSL}, the following relations hold for $\hCm$, $\hCp$ and $\bK$:
\begin{equation}\label{slCpmKP}
\begin{gathered}
\hCp+\hCm=\hC_{\ad},\qquad \bP\hC_{\pm}=\hC_{\pm}\bP=\pm\hC_{\pm},\qquad \hCp\hCm=\hCm\hCp=0,\\
\bK\hCm=\hCm\bK=0,\qquad \bK\hCp=\hCp\bK=-\bK,\\
\bK\hC_{\ad}=\hC_{\ad}\bK=-\bK.
\end{gathered}
\end{equation}
\begin{proposition}
The antisymmetric $\hCm$ and symmetric $\hCp$ parts of the split Casimir operator of the $s\ell(M|N)$ Lie superalgebra for $\omega\neq 0,1,2$ satisfy
\begin{equation}\label{slCmChar}
\hCm^2=-\frac{1}{2}\hCm\iff \hCm(\hCm+\frac{1}{2}\bI)=0.
\end{equation}
\begin{equation}\label{slCp3}
\hCp^3=-\frac{1}{2}\hCp^2+\frac{1}{\omega^2}\hCp+\frac{1}{4\omega^2}(\bI+\bPad-2\bK)
\end{equation}
 \begin{equation}\label{slCp4noK}
 \hCp^4=-\frac{3}{2}\hCp^3-\frac{\omega^2-2}{2\omega^2}\hCp^2+
 \frac{3}{2\omega^2}\hCp+\frac{1}{4\omega^2}(\bI+\bPad).
 \end{equation}
 \begin{equation}\label{slCpCharF}
\hCp(\hCp+\bI)(\hCp-\frac{1}{\omega}\bI)(\hCp+\frac{1}{\omega}\bI)(\hCp+\frac{1}{2}\bI)=0.
\end{equation}
The split Casimir operator $\hCad=\hCm+\hCp$ for $\omega\neq 0,1,2$ satisfies
\begin{equation}\label{slPolyF}
\hC_{\ad}(\hC_{\ad}+\bI)(\hC_{\ad}-\frac{1}{\omega}\bI)(\hC_{\ad}+\frac{1}{\omega}\bI)(\hC_{\ad}+\frac{1}{2}\bI)=0.
\end{equation}
\end{proposition}
\begin{proof}
Identity \eqref{slCmChar} for $s\ell(M|N)$ is a special case of \eqref{Cm2byCm} that holds for all Lie superalgebras with the nondegenerate Cartan-Killing metric.

By \eqref{kazSL} and \eqref{slKrel} we obtain for $\hCp^2$
\begin{equation}\label{slCp2}
\begin{aligned}
\hCp^2&=\frac{1}{8\omega^2}(I+\mcP_{13}\mcP_{24})(4I+2\mcK_{32}\mcK_{14}+
\omega\mcK_{32}+\omega\mcK_{14})\\
&-\frac{1}{4\omega^2}(\mcP_{13}+\mcP_{24})(\mcK_{12}+\mcK_{34}+\mcK_{32}+
\mcK_{14}).
\end{aligned}
\end{equation}
Multiplying \eqref{slCp2} by $\hCp$ yields \eqref{slCp3}. Multiplying once more \eqref{slCp3} by $\hCp$ and using \eqref{slCpmKP}, we obtain:
\begin{equation}\label{slCp4}
\hCp^4=-\frac{1}{2}\hCp^3+\frac{1}{\omega^2}\hCp^2+\frac{1}{2\omega^2}\hCp
+\frac{1}{2\omega^2}\bK.
\end{equation}
 Now we express $\bK$ from \eqref{slCp4} and substitute the result into \eqref{slCp3}, which gives \eqref{slCp4noK}. Multiplying both sides of \eqref{slCp4}
by $(\hCp + \bI)$ and using the last identity from \eqref{slCpmKP},
or multiplying both sides of \eqref{slCp4noK}
by $\hCp$ and using the second identity from \eqref{slCpmKP} yields
\begin{equation}\label{slCp5}
\hCp^5=-\frac{3}{2}\hCp^4-\frac{\omega^2-2}{2\omega^2}\hCp^3-
\frac{3}{2\omega^2}\hCp^2-\frac{1}{2\omega^2}\hCp,
\end{equation}
which can be rewritten as
\begin{equation}\label{slCpChar}
\hCp^5+\frac{3}{2}\hCp^4+\frac{\omega^2-2}{2\omega^2}\hCp^3-
\frac{3}{2\omega^2}\hCp^2-\frac{1}{2\omega^2}\hCp=0.
\end{equation}
Identity \eqref{slCpChar} is characteristic for $\hCp$.
The roots of the polynomial on the left-hand side of \eqref{slCpChar} are
\begin{equation}\label{slCpRoots}
a_1=0,\quad a_2=-1,\quad a_3=\frac{1}{\omega},\quad a_4=-\frac{1}{\omega},\quad a_5=-\frac{1}{2},
\end{equation}
hence the characteristic identity \eqref{slCpChar} takes the form \eqref{slCpCharF}. Note that for $\omega = 1,2$ we have degenerate roots, and these cases are considered separately (see below).

Since $\hCm=\hCm\bPad_-$ and $\hCp=\hCp\bPad_+$, from \eqref{slCmChar} and \eqref{slCpCharF} one infers
\begin{equation}\label{slCmCharM}
\hCm(\hCm+\frac{1}{2}\bPad_-)=0,
\end{equation}
 \begin{equation}\label{slCpCharFP}
\hCp(\hCp+\bPad_+)(\hCp-\frac{1}{\omega}\bPad_+)(\hCp+\frac{1}{\omega}\bPad_+)(\hCp+\frac{1}{2}\bPad_+)=0.
\end{equation}
These identities can be viewed as characteristic for $\hCm$ and $\hCp$ that are restricted to $\bPad_-(\Vad^{\otimes 2})$ and $\bPad_+(\Vad^{\otimes 2})$, respectively. In this sense, the roots of the characteristic polynomial on the right of \eqref{slCpCharFP} are
\begin{equation}\label{slCpRootsPr}
a'_1=-1,\quad a'_2=\frac{1}{\omega},\quad a'_3=-\frac{1}{\omega},\quad a'_4=-\frac{1}{2}.
\end{equation}

To find the characteristic polynomial of $\hC_{\ad}$, we substitute $\hC_{\ad}$ for $\hCp$ in \eqref{slCpCharF} and use $\hC_{\ad}=\hCm+\hCp$ and $\hCp\hCm=0$:
\begin{equation}
\begin{gathered}
\hC_{\ad}(\hC_{\ad}+\bI)(\hC_{\ad}-\frac{1}{\omega}\bI)(\hC_{\ad}+\frac{1}{\omega}\bI)(\hC_{\ad}+\frac{1}{2}\bI)=\\
=\hCp(\hCp+\bI)(\hCp-\frac{1}{\omega}\bI)(\hCp+\frac{1}{\omega}\bI)(\hCp+\frac{1}{2}\bI)+\\
+\hCm(\hCm+\bI)(\hCm-\frac{1}{\omega}\bI)(\hCm+\frac{1}{\omega}\bI)(\hCm+\frac{1}{2}\bI)=0,
\end{gathered}
\end{equation}
where the last relation holds by \eqref{slCmChar} and \eqref{slCpCharF}. Therefore, the characteristic identity for $\hC_{\ad}$ is given by \eqref{slPolyF}.
\end{proof}
Note that the roots of the polynomial on the left of \eqref{slPolyF} coincide with \eqref{slCpRoots}, and for $\omega=1, 2$ we have degenerate roots:
\begin{equation}\label{slPolyDegR}
\omega=1  \implies a_2=a_4=-1,\qquad\qquad
\omega=2  \implies a_4=a_5=-\frac{1}{2}
\end{equation}
Leaving in the polynomial on the left of \eqref{slPolyF} only one of the parentheses that correspond to the roots \eqref{slPolyDegR}, we get for $\omega= 1$:
\begin{equation}
\hC_{\ad}(\hC_{\ad}+\bI)(\hC_{\ad}-\bI)(\hC_{\ad}+\frac{1}{2}\bI)=\frac{1}{2}\bK\neq 0.
\end{equation}
Therefore, for $\omega= 1$ the $s\ell(M|N)$ representation $\ad^{\otimes 2}$ is not completely reducible. For $\omega=2$, we analogously get
\begin{equation}\label{slo2}
\hC_{\ad}(\hC_{\ad}+\bI)(\hC_{\ad}-\frac{1}{2}\bI)(\hC_{\ad}+\frac{1}{2}\bI)=\frac{1}{16}(\bPad_++\bK)-\frac{1}{4}\hCp^2.
\end{equation}
Using the component form of \eqref{slo2}, one can check that for $M=2$, $N=0$ this expression is nullified, while for $M=3$, $N=1$ it does not. Therefore, the representation $\ad^{\otimes 2}$ of $s\ell(M|N)$ for $\omega= 2$ in general is not completely reducible. The exceptional cases $\omega\neq 1,2$ are to be considered later in this section.

To construct projectors onto invariant subspaces of the representation $\ad^{\otimes 2}$ of $s\ell(M|N)$, we utilize the fact that for an arbitrary Lie superalgebra the symmetric and antisymmetric parts of $\Vad^{\otimes 2}$ are invariant (see Section \ref{21} and \eqref{XPcom}). The symmetric invariant subspaces $\bPad_+(\Vad^{\otimes 2})$ of $\Vad^{\otimes 2}$ can be expressed as eigenspaces of $\hCp$, which can be viewed as acting in $\bPad_+(\Vad^{\otimes 2})$. The role of the identity operator here is played by $\bPad_+$. By \eqref{GenPr}, where we suppose $p=5$, $\hC_T=\hCp$, $I_T^{\otimes 2}=\bPad_+$, and $a_i$ are given in \eqref{slCpRootsPr},
\begin{equation}\label{slProjSym}
\begin{aligned}
\proj_{a'_1}^{(+)}\equiv\proj_1^{(+)}&=\frac{1}{\omega^2-1}\bK,\\
\proj_{a'_2}^{(+)}\equiv\proj_2^{(+)}&=
-\frac{\omega}{2(\omega+1)(\omega+2)}\bK+
\frac{\omega^2}{\omega+2}\hCp^2+\frac{\omega}{2}\hCp+
\frac{\omega}{2(\omega+2)}\bPad_+,\\
\proj_{a'_3}^{(+)}\equiv\proj_3^{(+)}&=\frac{\omega}{2(\omega-1)(\omega-2)}\bK-
\frac{\omega^2}{\omega-2}\hCp^2-\frac{\omega}{2}\hCp+
\frac{\omega}{2(\omega-2)}\bPad_+,\\
\proj_{a'_4}^{(+)}\equiv\proj_4^{(+)}&=
\frac{4}{\omega^2-4}\big(\omega^2\hCp^2-
\bPad_+-\bK\big).
\end{aligned}
\end{equation}
One can easily check that $\proj_1^{(+)}+\proj_2^{(+)}+\proj_3^{(+)}+\proj_4^{(+)}=\bPad_+$.

For $\omega=1$, the characteristic identity for $\hCp$ is
\begin{equation}
(\hCp+\bI)^2(\hCp-\bI)(\hCp+\frac{1}{2}\bI)\bPad_+=0.
\end{equation}
Therefore, in \eqref{charct2} one needs to put $a'_1=-1,\ a'_2=1,\ a'_3=-\frac{1}{2},\ k_1=2,\ k_2=k_3=1$. Similarly to the considered above general case, we build projectors onto symmetric generalized eigenspaces of $\hCp$. By \eqref{GenPr2}, where we substitute $\bPad_+$ for $\bI^{\otimes 2}_T$,
\begin{equation}
\begin{aligned}
\proj_{a'_1}^{(+)}\equiv\proj_1^{(+)}&=-\frac{1}{2}\bPad_+-\frac{5}{4}\bK-\frac{1}{2}\hCp+\hCp^2,\\
\proj_{a'_2}^{(+)}\equiv\proj_2^{(+)}&=\frac{1}{6}\bPad_+-\frac{1}{12}\bK+\frac{1}{2}\hCp+\frac{1}{3}\hCp^2,\\
\proj_{a'_3}^{(+)}\equiv\proj_3^{(+)}&=\frac{4}{3}\bPad_++\frac{4}{3}\bK-\frac{4}{3}\hCp^2.
\end{aligned}
\end{equation}
Of these operators, $\proj^{(+)}_2$ and $\proj^{(+)}_3$ project onto the eigenspaces of $\hCp$, while $\proj_1^{(+)}$ extracts its generalized eigenspace. Thus, the action of $s\ell(M|N)$ in $\proj_1^{(+)}(\Vad^{\otimes 2})$ for $\omega=1$ is reducible but not completely reducible. It is related to the fact that $\bK$, satisfying $\bK^2=0$, is nilpotent and thus not diagonalizable in this case.

For $\omega=2$, the characteristic identity for $\hCp$ is
\begin{equation}
(\hCp+\bI)(\hCp-\frac{1}{2}\bI)(\hCp+\frac{1}{2}\bI)^2\bPad_+=0,
\end{equation}
so in \eqref{charct2} we put $a'_1=-1,\ a'_2=\frac{1}{2},\ a'_3=-\frac{1}{2},\ k_1=k_2=1,\ k_3=2$. The projectors onto the generalized eigenspaces of $\hCp$ are given by \eqref{GenPr2}:
\begin{equation}
\begin{aligned}
\proj_{a'_1}^{(+)}\equiv\proj_1^{(+)}&=\frac{1}{3}\bK,\\
\proj_{a'_2}^{(+)}\equiv\proj_2^{(+)}&=\frac{1}{4}\bPad_+-\frac{1}{12}\bK+\hCp+\hCp^2,\\
\proj_{a'_3}^{(+)}\equiv\proj_3^{(+)}&=\frac{3}{4}\bPad_+-\frac{1}{4}\bK-\hCp-\hCp^2.
\end{aligned}
\end{equation}
Of these operators, $\proj_1^{(+)}$ and $\proj_2^{(+)}$ project onto the eigenspaces of $\hCp$ while the image of $\proj_3^{(+)}$ is a generalized eigenspace of $\hCp$. Thus, the restriction of the representation $\ad^{\otimes 2}$ to $\proj_3^{(+)}(\Vad^{\otimes 2})$ is neither irreducible nor completely reducible for $\omega\equiv M-N=2$. For simplicity, in what follows we assume $\proj_4^{(+)}=0$ for $\omega=1,2$.

To find projectors onto antisymmetric invariant subspaces of the representation $\ad^{\otimes 2}$ of $s\ell(M|N)$ that can be expressed as eigenspaces of $\hCm$ (which is viewed here as acting in $\bPad_-(\Vad^{\otimes 2})$, where the role of the identity operator is played by $\bPad_-$), we use \eqref{GenPr} where, in accordance with \eqref{slCmCharM}, we put $p=2$, $\hC_T=\hCm$, $I_T^{\otimes 2}=\bPad_-$, $a_1=0$ and $a_2=-\frac{1}{2}$:
\begin{equation}\label{slProjAsym1}
\proj_{a_1}^{(-)}\equiv\proj_1^{(-)}=2\hCm+\bPad_-,\qquad\qquad
\proj_{a_2}^{(-)}\equiv\proj_2^{(-)}=-2\hCm.
\end{equation}
Apparently, $\proj_1^{(-)}+\proj_2^{(-)}=\bPad_-$, i.e. $\proj_1^{(-)}+\proj_2^{(-)}+\proj_2^{(+)}+\proj_3^{(+)}+\proj_4^{(+)}=\bPad_-+\bPad_+=\bI$.

As a result, we have the following full system of projectors: $\proj_1^{(-)},\ \proj_2^{(-)},\ \proj_2^{(+)},\ \proj_3^{(+)},\ \proj_4^{(+)},\ \proj_5^{(+)}$. However, not all of those projectors are primitive. In order to show this, we use the following relation, which holds for the generators $T_{ij}$ of $s\ell(M|N)$ in the defining representation $T_f$ (we denote $T_f(T_{ij})=T_{ij}$):
\begin{equation}
T_{ij}T_{km}+(-1)^{([i]+[j])([k]+[m])}T_{km}T_{ij}\equiv\lbrack T_{ij},T_{km}\rbrack_+=D^{rs}{}_{ij,km}T_{rs}+\alpha{\sf g}_{ij,km}.
\end{equation}
Here $\alpha=\frac{2c_2(T_f)}{\sdim(s\ell(M|N))}=\frac{1}{\omega^2}$ where $c_2(T_f)=\frac{\omega^2-1}{2\omega^2}$ is the value of the quadratic Casimir operator \eqref{C2G} in the defining representation $T_f$, and the numbers $D^{rs}{}_{ij,km}$ are called the symmetric in pairs of indices $(ij)$ and $(km)$ structure constants of $s\ell(M|N)$. Define the operators $\tC$ and $\tCm$, that act on $V_{\ad}^{\otimes 2}$:
\begin{equation}\label{tCmdef}
\begin{gathered}
\tC^{i_1i_2i_3i_4}{}_{j_1j_2j_3j_4}=(-1)^{([j_1]+[j_2])([b_1]+[b_2])}\okg^{a_1a_2,b_1b_2}X^{i_1i_2}{}_{a_1a_2,j_1j_2}D^{i_3i_4}{}_{b_1b_2,j_3j_4},\\
\tCm=\frac{\omega}{4}(\bI-\bPad)\tC(\bI-\bPad),
\end{gathered}
\end{equation}
where $X^{i_1i_2}{}_{a_1a_2,j_1j_2}$ are the structure constants of $s\ell(M|N)$, given in \eqref{slStruc}, and $\okg^{a_1a_2,b_1b_2}$ is the inverse Cartan-Killing metric \eqref{InvKilSl}.

\begin{proposition}
The explicit form of $\tCm$ defined in \eqref{tCmdef} and acting in $\Vad^{\otimes 2}$, is
\begin{equation}\label{tCmExp}
\tCm=\frac{1}{2}(\mcP_{13}-\mcP_{24})\big(I-\frac{1}{\omega}(\mcK_{12}+\mcK_{34}+\mcK_{32}+\mcK_{14})\big).
\end{equation}
Besides, $\tCm$ satisfies
\begin{equation}\label{tCmRels}
\begin{gathered}
\bPad_+\tCm=\tCm\bPad_+=0,\qquad \tCm\hCm=\hCm\tCm=0,\\
\tCm^2=2\hCm+\bPad_-,\qquad \tCm(\tCm+\bI)(\tCm-\bI)=0.
\end{gathered}
\end{equation}
The last identity in the second row in \eqref{tCmRels} is characteristic for $\tCm$.
\end{proposition}
\begin{proof}
Direct calculations show that the symmetric structure constants $D^{rs}{}_{ij,km}$ of $s\ell(M|N)$ satisfy
\begin{gather}
D^{rs}{}_{ij,km}=(\delta^r_l\delta^s_n-\frac{1}{\omega}\mcK^{rs}{}_{ln})\bar{D}_{ij,km}{}^{ln},\\
\bar{D}^{rs}{}_{ij,km}=\delta^r_i\delta^s_m\delta_{jk}+(-1)^{([i]+[j])([k]+[m])}\delta^r_k\delta^s_j\delta_{im}-\frac{2}{\omega}\big((-1)^{[i]}\delta^r_k\delta^s_m\delta_{ij}+(-1)^{[m]}\delta^r_i\delta^s_j\delta_{km}\big).\nonumber
\end{gather}
Thus, the operator $\tC$ given in \eqref{tCmdef} equals
\begin{equation}
\tC=(I-\frac{1}{\omega}\mcK_{34})\big(\mcP_{13}-\mcP_{24}+\mcK_{14}-\mcK_{32}+\frac{2}{\omega}(\mcP_{24}-\mcP_{13})\mcK_{34}\big)
\end{equation}
Using \eqref{tCmdef} and the equalities $(\bI-\bPad)=(I-\bP)\oI_{12}\oI_{34}$, $\bP=\mcP_{13}\mcP_{24}$, we get \eqref{tCmExp}. One can check \eqref{tCmRels} by direct calculations using \eqref{tCmExp}, \eqref{slhCpmExp}, \eqref{slPadDif}.
\end{proof}
Since $\tCm=\tCm\bPad_-=\bPad_-\tCm$, the last equality in \eqref{tCmRels} can be rewritten as
\begin{equation}\label{tCmchar}
\tCm(\tCm+\bPad_-)(\tCm-\bPad_-)=0.
\end{equation}
It is the characteristic identity for $\tCm$, which is restricted to the antisymmetric part $\bPad_-(\Vad^{\otimes 2})$ of $\Vad^{\otimes 2}$. Using \eqref{tCmchar} and \eqref{GenPr}, one can obtain projectors onto the eigenspaces of $\tCm$. The explicit formulas are
\begin{equation}\label{slTildePr}
\widetilde{\proj}_0^{(-)}=-2\hCm,\qquad
\widetilde{\proj}_{-1}^{(-)}=\hCm+\frac{1}{2}\bPad_--\frac{1}{2}\tCm,\qquad
\widetilde{\proj}_{1}^{(-)}=\hCm+\frac{1}{2}\bPad_-+\frac{1}{2}\tCm.
\end{equation}
In \eqref{slTildePr}, the lower indices of the projectors equal the eigenvalues of $\tCm$ in the corresponding eigenspaces.

Note that $\widetilde{\proj}^{(-)}_{-1}+\widetilde{\proj}^{(-)}_1=\proj_1^{(-)}$ where $\proj_1^{(-)}$ is given in \eqref{slProjAsym1}, and $\widetilde{\proj}^{(-)}_0=\proj_2^{(-)}$ where $\proj_2^{(-)}$ is defined in \eqref{slProjAsym1}. Thus, we get the following full system of mutually orthogonal projectors for $\omega\neq 0,1,2$:
\begin{equation}\label{slProj}
\begin{aligned}
\widetilde{\proj}^{(-)}_{-1}&=\hCm+\frac{1}{2}\bPad_--\frac{1}{2}\tCm,\qquad
\widetilde{\proj}^{(-)}_{1}=\hCm+\frac{1}{2}\bPad_-+\frac{1}{2}\tCm,\qquad \proj_2^{(-)}=-2\hCm,\\
\proj_1^{(+)}&=\frac{1}{\omega^2-1}\bK,\\
\proj_2^{(+)}&=-\frac{\omega}{2(\omega+1)(\omega+2)}\bK+\frac{\omega^2}{\omega+2}\hCp^2+\frac{\omega}{2}\hCp+\frac{\omega}{2(\omega+2)}\bPad_+,\\
\proj_3^{(+)}&=\frac{\omega}{2(\omega-1)(\omega-2)}\bK-\frac{\omega^2}{\omega-2}\hCp^2-\frac{\omega}{2}\hCp+\frac{\omega}{2(\omega-2)}\bPad_+,\\
\proj_4^{(+)}&=\frac{4}{\omega^2-4}\big(\omega^2\hCp^2-\bPad_+-\bK\big).
\end{aligned}
\end{equation}
The images of $\widetilde{\proj}_{\pm 1}^{(-)}$ and $\proj_2^{(-)}$ lie within the antisymmetric part $\bPad_-(\Vad^{\otimes 2})$ of the space $\Vad^{\otimes 2}$, while the images of $\proj_i^{(+)}$, $(i=2,...,5)$ belong to its symmetric part $\bPad_+(\Vad^{\otimes 2})$.

For $\omega=1,2$ the projectors are given by (the left and right columns correspond to $\omega=1$ and $\omega=2$, respectively)
 \begin{equation*}
 \begin{aligned}
 \widetilde{\proj}^{(-)}_{-1}&=\hCm+\frac{1}{2}\bPad_--\frac{1}{2}\tCm, & \qquad\qquad \widetilde{\proj}^{(-)}_{-1}&=\hCm+\frac{1}{2}\bPad_--\frac{1}{2}\tCm\\
\widetilde{\proj}^{(-)}_{1}&=\hCm+\frac{1}{2}\bPad_-+\frac{1}{2}\tCm, & \widetilde{\proj}^{(-)}_{1}&=\hCm+\frac{1}{2}\bPad_-+\frac{1}{2}\tCm,\\
\proj_2^{(-)}&=-2\hCm, & \proj_2^{(-)}&=-2\hCm\\
\proj_1^{(+)}&=-\frac{1}{2}\bPad_+-\frac{5}{4}\bK-\frac{1}{2}\hCp+\hCp^2, & \proj_1^{(+)}&=\frac{1}{3}\bK,\\
\proj_2^{(+)}&=\frac{1}{6}\bPad_+-\frac{1}{12}\bK+\frac{1}{2}\hCp+\frac{1}{3}\hCp^2, & \proj_2^{(+)}&=\frac{1}{4}\bPad_+-\frac{1}{12}\bK+\hCp+\hCp^2,\\
\proj_3^{(+)}&=\frac{4}{3}\bPad_++\frac{4}{3}\bK-\frac{4}{3}\hCp^2, & \proj_3^{(+)}&=\frac{3}{4}\bPad_+-\frac{1}{4}\bK-\hCp-\hCp^2.
 \end{aligned}
 \end{equation*}

In order to find the dimensions of the invariant subspaces extracted by the projectors \eqref{slProj}, we calculate the traces and supertraces of those projectors. First, we compute some auxiliary traces and supertraces (recall that $\xi=M+N$):
\begin{equation}\label{slAuxTr}
\begin{aligned}
\tr \bI&=(\xi^2-1)^2, & \str\bI&=(\omega^2-1)^2,\\
\tr \bPad_+&=\frac{1}{2}(\xi^2-1)^2+\frac{1}{2}(\omega^2-1), &\qquad \str\bPad_+&=\frac{1}{2}\omega^2(\omega^2-1),\\
\tr\bPad_-&=\frac{1}{2}(\xi^2-1)^2-\frac{1}{2}(\omega^2-1), & \str \bPad_-&=\frac{1}{2}(\omega^2-1)(\omega^2-2),\\
\tr\bK&=\omega^2-1, & \str\bK&=\omega^2-1,\\
\tr\hCp&=\frac{1}{2}(\xi^2-1), & \str\hCp&=\frac{1}{2}(\omega^2-1),\\
\tr\hCp^2&=\frac{\xi^4}{2\omega^2}+\frac{\xi^2}{4}-2\frac{\xi^2}{\omega^2}+\frac{5}{4},& \str\hCp^2&=\frac{3}{4}(\omega^2-1),\\
\tr\hCm&=-\frac{1}{2}(\xi^2-1),& \str\hCm&=-\frac{1}{2}(\omega^2-1),\\
\tr\tCm&=0, & \str\tCm&=0.
\end{aligned}
\end{equation}
Using \eqref{slProj} and \eqref{slAuxTr}, we obtain the traces
\begin{align}
\tr\widetilde{\proj}^{(-)}_{-1}&=\frac{1}{4}\big((\xi^2-2)^2-\omega^2\big), \qquad& \tr\proj_1^{(+)}&=1,\\
\tr\widetilde{\proj}^{(-)}_{1}&=\frac{1}{4}\big((\xi^2-2)^2-\omega^2\big), & \tr\proj_2^{(+)}&=\frac{1}{4}\big((\xi^2-1)^2+2(\xi^2+1)(\omega-1)+(\omega-1)^2\big),\nonumber\\
\tr\proj_2^{(-)}&=\xi^2-1, & \tr\proj_3^{(+)}&=\frac{1}{4}\big((\xi^2-1)^2-2(\xi^2+1)(\omega+1)+(\omega+1)^2\big),\nonumber\\
&{} & \tr\proj_4^{(+)}&=\xi^2-1\nonumber
\end{align}
and supertraces of the projectors \eqref{slProj}:
\begin{equation}
\begin{aligned}
\str\widetilde{\proj}^{(-)}_{-1}&=\frac{1}{4}(\omega^2-1)(\omega^2-4), \qquad\qquad& \str\proj_1^{(+)}&=1,\\
\str\widetilde{\proj}^{(-)}_{1}&=\frac{1}{4}(\omega^2-1)(\omega^2-4), & \str\proj_2^{(+)}&=\frac{1}{4}\omega^2(\omega-1)(\omega+3),\\
\str\proj_2^{(-)}&=\omega^2-1, & \str\proj_3^{(+)}&=\frac{1}{4}\omega^2(\omega+1)(\omega-3),\\
&{} & \str\proj_4^{(+)}&=\omega^2-1
\end{aligned}
\end{equation}
for $\omega\neq 0,1,2$. Analogously to the $osp(M|N)$ case, we find the dimensions of even and odd parts of the invariant subspaces:
\begin{align}
\dim_{\oo}\widetilde{V}^{(-)}_{-1}&=\frac{1}{4}(M^2-1)(M^2-4)+
\frac{1}{4}(N^2-1)(N^2-4)+\nonumber\\
&+\frac{1}{2}(MN+1)(3MN-2),\nonumber\\
\dim_{\oo}\widetilde{V}^{(-)}_{1}&=\frac{1}{4}(M^2-1)(M^2-4)+
\frac{1}{4}(N^2-1)(N^2-4)+\nonumber\\
&+\frac{1}{2}(MN+1)(3MN-2),\nonumber\\
\dim_{\oo}V_2^{(-)}&=M^2+N^2-1,\nonumber\\
\dim_{\oo}V_1^{(+)}&=1,\label{slEvDim} \\ 
\dim_{\oo}V_2^{(+)}&=\frac{1}{4}M^2(M-1)(M+3)+\frac{1}{4}N^2(N+1)(N-3)+\nonumber\\
&+\frac{1}{2}MN(3MN-M+N-1),\nonumber\\
\dim_{\oo}V_3^{(+)}&=\frac{1}{4}M^2(M+1)(M-3)+\frac{1}{4}N^2(N-1)(N+3)+\nonumber\\
&+\frac{1}{2}MN(3MN+M-N-1),\nonumber\\
\dim_{\oo}V_4^{(+)}&=M^2+N^2-1,
\end{align}
\begin{align}
\dim_{\ol}\widetilde{V}^{(-)}_{-1}&=MN(M^2+N^2-2), & \dim_{\ol}V_1^{(+)}&=0,\label{slOdDim}\\
\dim_{\ol}\widetilde{V}^{(-)}_{1}&=MN(M^2+N^2-2), & \dim_{\ol}V_2^{(+)}&=MN\big(M(M+1)+N(N-1)-2\big),\nonumber\\
\dim_{\ol}V_2^{(-)}&=2MN, & \dim_{\ol}V_3^{(+)}&=MN\big(M(M-1)+N(N+1)-2\big),\nonumber\\
&{} & \dim_{\ol}V_4^{(+)}&=2MN.\nonumber
\end{align}
Note that substitutions $M=0$ and $N=0$ nullify the dimensions of the odd parts of $\ad^{\otimes 2}$-invariant subspaces of $s\ell(M|N)$, as it should be. Besides, the mentioned substitutions turn \eqref{slEvDim} into the corresponding expressions for the dimensions of the invariant subspaces of the $s\ell(N)$ (or $s\ell(M)$) Lie algebra, which are given in \cite{IsKr}.

\section{Universal characteristic identities for $\hCp$ in the cases of the $osp(M|N)$ and $s\ell(M|N)$ Lie superalgebras}\label{Sec5}

For the $osp(M|N)$ and $s\ell(M|N)$ Lie superalgebras (which are denoted by $\mfg$ in this section), the characteristic identities \eqref{Cp3} and \eqref{slCp3} for $\hCp$ in the adjoint representation can be written in the following general form:
\begin{equation}\label{hCpUni3}
\hCp^3+\frac{1}{2}\hCp^2=\mu_1\hCp+\mu_2(\bI+\bPad-2\bK) \; ,
\end{equation}
which coincides precisely with the form of the universal identity for $\hCp$ in the case of the classical Lie algebras \cite{IsKr}.
The parameters $\mu_1$ and $\mu_2$ corresponding to the algebras $osp(M|N)$ and $s\ell(M|N)$ are given in Table $1$.
 \vspace{-0.5cm}
{\renewcommand{\arraystretch}{1.3}
\begin{table}[H]\label{mu1mu2init}
\centering
\caption{The values of $\mu_1$ and $\mu_2$ for the $osp(M|N)$ and $s\ell(M|N)$ Lie superalgebras} \vspace{0.2cm}
\begin{tabular}{|c|c|c|}
\hline
 & $\mu_1$ & $\mu_2$\\
 \hline
$osp(M|N)$ & $-\frac{\omega-8}{2(\omega-2)^2}$ & $\frac{\omega-4}{2(\omega-2)^3}$\\
\hline
$s\ell(M|N)$ & $\frac{1}{\omega^2}$ & $\frac{1}{4\omega^2}$\\
\hline
\end{tabular}
\end{table}

The subsequent analysis of \eqref{hCpUni3} mostly follows the consideration \cite{IsKr} of an analogous identity for the classical Lie algebras.
Multiplying both sides of \eqref{hCpUni3} by $\bK$ and using
\begin{equation}
\bK(\bI+\bPad)=2\bK,\qquad \bK\hCp=-\bK,\qquad \bK\cdot \bK=\sdim\mfg\bK,
\end{equation}
one can express the superdimension of $\mfg$ as a function of $\mu_1$ and $\mu_2$:
\begin{equation}\label{dimmu}
\sdim\mfg=\frac{2\mu_2-\mu_1+\frac{1}{2}}{2\mu_2}.
\end{equation}
Multiplying then both sides of \eqref{hCpUni3} by $\hCp(\hCp+\bI)$, we obtain the characteristic identity for $\hCp$:
\begin{equation}\label{Vogmu}
\hCp(\hCp+\bI)(\hCp^3+\frac{1}{2}\hCp^2-\mu_1\hCp-2\mu_2\bI)=0,
\end{equation}
the factorized form of which is
\begin{equation}\label{Vogabg}
\hCp(\hCp+\bI)(\hCp+\frac{\alpha}{2t}\bI)
(\hCp+\frac{\beta}{2t}\bI)(\hCp+\frac{\gamma}{2t}\bI)=0.
\end{equation}
Thus, the roots of the polynomial on the left of \eqref{Vogmu} are
\begin{equation}\label{VogRoots}
a_1=0,\qquad a_2=-1, \qquad a_3=-\frac{\alpha}{2t},\qquad a_4=-\frac{\beta}{2t},\qquad a_5=-\frac{\gamma}{2t},
\end{equation}
where the normalisation parameter $t=\alpha+\beta+\gamma$, as by \eqref{Vogmu} and
\eqref{Vogabg},
\begin{equation}
\frac{\alpha}{2t}+\frac{\beta}{2t}+\frac{\gamma}{2t}= \frac{1}{2},
\end{equation}
We choose $t=h^\lor$ where $h^\lor$ is the dual Coxeter number of $\mfg$. The values of $h^\lor$ (see, e.g., \cite{KacWak}) for $s\ell(M|N)$ and $osp(M|N)$ are presented in Table \ref{CxtrNum}. As usual, $\omega=2m+1-N$ and $\omega=2m-N$ for $\mfg=osp(2m+1|N)$ and $\mfg=osp(2m,N)$, respectively.
{\begin{table}[H]
\centering
\caption{The dual Coxeter numbers for the Lie superalgebras $s\ell(M|N)$ and $osp(M|N)$}\label{CxtrNum} \vspace{0.2cm}
\begin{tabular}{|c|c|c|c|}
\hline
& $s\ell(M|N)$ & $\begin{array}{c} osp(2m+1|N),\ \omega>1 \\ osp(2m|N),\ \omega>0 \end{array}$ & $\begin{array}{c} osp(2m+1|N),\ \omega\le 1\\osp(2m|N),\ \omega\le 0 \end{array}$\\
\hline
$h^\lor$ & $\omega$ & $\omega-2$ & $-\frac{1}{2}(\omega-2)$\\
\hline
\end{tabular}
\end{table}}
The parameters $\alpha$, $\beta$ and $\gamma$ were introduced by Vogel in \cite{Vog}. The values of these parameters for the algebras $osp(M|N)$ and $s\ell(M|N)$ can be found from \eqref{Cp5F} and \eqref{slCpCharF} and are given in Table \ref{abgVal}.\\
\vspace{-0.5cm}
{\begin{table}[H]
\centering
\caption{The Vogel parameters for the $osp(M|N)$ and $s\ell(M|N)$ Lie superalgebras}\label{abgVal} \vspace{0.2cm}
\begin{tabular}{|c|c|c|c|}
\hline
& $s\ell(M|N)$ & $\begin{array}{c} osp(2m+1|N),\ \omega>1 \\ osp(2m|N),\ \omega>0 \end{array}$ & $\begin{array}{c} osp(2m+1|N),\ \omega\le 1\\osp(2m|N),\ \omega\le 0 \end{array}$\\
 \hline
 $\alpha$ & $-2$ & $-2$ & $1$ \\
 \hline
 $\beta$ & $2$ & $4$ & $-2$\\
 \hline
 $\gamma$ & $\omega $ & $\omega-4$ & $-\frac{1}{2}(\omega-4)$\\
 \hline
 $t$ & $\omega$ & $\omega-2$ & $-\frac{1}{2}(\omega-2)$\\
 \hline
\end{tabular}
\end{table}}
A comparison of \eqref{Vogmu} and \eqref{Vogabg} shows that $\mu_1$ and $\mu_2$ can be expressed in terms of the Vogel parameters:
\begin{equation}
\mu_1=-\frac{\alpha\beta+\alpha\gamma+\beta\gamma}{4t^2},\qquad \mu_2=-\frac{\alpha\beta\gamma}{16t^3},
\end{equation}
while the superdimension \eqref{dimmu} of $\mfg$ acquires the universal form
\begin{equation}\label{dimvog}
\sdim\mfg=\frac{(\alpha-2t)(\beta-2t)(\gamma-2t)}{\alpha\beta\gamma}.
\end{equation}
Now using \eqref{Vogabg}, we can obtain a universal form of the projectors $\proj^{(+)}_{(a_i)}$ onto the invariant subspaces $V_{(a_i)}$ of the symmetric space $\bP_+(V_{\ad}^{\otimes 2})$:
\begin{equation}\label{vogproj}
\begin{gathered}
\proj^{(+)}_{(-\frac{\alpha}{2t})}= \proj^{(+)}(\alpha|\beta,\gamma),\qquad \proj^{(+)}_{(-\frac{\beta}{2t})}=\proj^{(+)}(\beta|\alpha,\gamma),\qquad \proj^{(+)}_{(-\frac{\gamma}{2t})}=\proj^{(+)}(\gamma|\alpha,\beta)\\
\proj^{(+)}_{(-1)}=\frac{1}{\sdim \mfg}\bK,
\end{gathered}
\end{equation}
where we denoted
\begin{equation}\label{vogprojabg}
\proj^{(+)}(\alpha|\beta,\gamma)=
\frac{4t^2}{(\beta-\alpha)(\gamma-\alpha)}
\Big(\hCp^2+\big(\frac{1}{2}-\frac{\alpha}{2t}\big)\hCp+
\frac{\beta\gamma}{8t^2}\big(\bI+\bPad-\frac{2\alpha}{\alpha-2t}
\bK\big)\Big).
\end{equation}
By \eqref{strGenMany} and \eqref{dimvog}, the supertrace of $\proj^{(+)}(\alpha|\beta,\gamma)$ is
\begin{equation}\label{strPabg}
\str\proj^{(+)}(\alpha|\beta,\gamma)=-\frac{(3\alpha-2t)(\beta-2t)(\gamma-2t)(\beta+t)(\gamma+t)t}{\alpha^2(\alpha-\beta)(\alpha-\gamma)\beta\gamma}.
\end{equation}
From \eqref{strPabg} and \eqref{dimvog} we get the superdimensions of the invariant subspaces $V_{(-1)}$, $V_{(-\frac{\alpha}{2t})}$, $V_{(-\frac{\beta}{2t})}$ and $V_{(-\frac{\gamma}{2t})}$ extracted by the projectors \eqref{vogproj}:
\begin{equation}
\begin{aligned}
\sdim V_{(-1)}&=\str \proj^{(+)}_{(-1)}=1,\\
\sdim V_{(-\frac{\alpha}{2t})}&=\str\proj^{(+)}_{(-1)}=-\frac{(3\alpha-2t)(\beta-2t)(\gamma-2t)(\beta+t)(\gamma+t)t}{\alpha^2(\alpha-\beta)(\alpha-\gamma)\beta\gamma},\\
\sdim V_{(-\frac{\beta}{2t})}&=\str\proj^{(+)}_{(-1)}=-\frac{(3\beta-2t)(\alpha-2t)(\gamma-2t)(\alpha+t)(\gamma+t)t}{\beta^2(\beta-\alpha)(\beta-\gamma)\alpha\gamma},\\
\sdim V_{(-\frac{\gamma}{2t})}&=\str\proj^{(+)}_{(-1)}=
-\frac{(3\gamma-2t)(\beta-2t)(\alpha-2t)(\beta+t)(\alpha+t)t}{
\gamma^2(\gamma-\beta)(\gamma-\alpha)\beta\alpha}.
\end{aligned}
\end{equation}
It is worth noting that for the case of Lie algebras nullification of either $3\alpha-2t$, $3\beta-2t$, or $3\gamma-2t$ corresponds to the exceptional Lie algebras $g_2$, $f_4$, $e_6$, $e_7$, $e_8$ as well as $s\ell(3)$ and $so(8)$. In this sense, the ``exceptional'' basic classical Lie superalgebras (see their definition, e.g., in \cite{Sorba2}) are $s\ell(M|N)$ for $M-N=0,\pm 3$, $osp(M|N)$ for $M-N=-1,8$ and $F(4)$.
\section{Eigenvalues of higher Casimir operators in the adjoint representation}\label{Sec6}
In this section, we derive a formula that expresses the eigenvalues of the higher Casimir operators in the adjoint representation in terms of the Vogel parameters.
Our method of constructing higher Casimir operators for Lie superalgebras is based on the method proposed in \cite{Okub} for Lie algebras (see also \cite{Book1}).

Let ${Y_A}$ be a homogeneous basis of the enveloping algebra $\mcU(\mfg)$ of the Lie superalgebra $\mfg$. If $\widetilde{C}=D^{AB}Y_A\otimes Y_B\in\mcU(\mfg)\otimes \mcU(\mfg)$ is an $\ad$-invariant operator, then for an arbitrary representation $T:\mfg\to (\End(V))_L$ of $\mfg$ the operator
\begin{equation}\label{hCtoC}
C=\str_2(({\rm id}
\otimes T) \; \widetilde{C})= D^{AB}Y_A\str(T(Y_B))
\end{equation}
lies in the centre of $\mcU(\mfg)$. Here
${\rm id}$ is the identity operator and
$\str_2$ denotes the supertrace in the second factor in $\mcU(\mfg)\otimes ((\End(V))_L$. In what follows, we are only interested in the operator $C$ in a particular representation $T'$. For $T'(C)$ we get:
\begin{equation}\label{hCC}
T'(C)=\str_2\bigl( (T'\otimes T) \; \widetilde{C} \bigr)=D^{AB}T'(Y_A)\str(T(Y_B)).
\end{equation}

Apparently, for the simple Lie superalgebra $\mfg$ where $\dim(\mfg)=n$ and the Cartan-Killing metric ${\sf g}_{ab}$ of which is nondegenerate, an arbitrary power $\hC^k$ of $\hC$ defined in \eqref{SpCas} is ad-invariant. The explicit form of $\hC^k$ is:
\begin{equation}
\hC^k=(-1)^{\sum_{i>j}^n[a_i][a_j]}{\sf g}^{a_1b_1}\dots {\sf g}^{a_nb_n}X_{a_1}\cdots X_{a_n}\otimes X_{b_1}\cdots X_{b_n}.
\end{equation}
Substituting $\widetilde{C}=\hC^k$ and $T'=T=\ad$ into \eqref{hCC} yields a relation for the $k$-th Casimir operator $\ad(C_k)$ in the adjoint representation:
\begin{equation}\label{Ckad}
\ad(C_k)=\str_2(\ad^{\otimes 2}(\hC^k))\equiv\str_2(\hC_{\ad}^k)
={\sf g}^{a_1\dots a_n}\ad(X_{a_1})\cdots \ad(X_{a_n}),
\end{equation}
where we denoted
\begin{equation}
{\sf g}^{a_1\dots a_2}=(-1)^{\sum_{i>j}^n[a_i][a_j]}{\sf g}^{a_1b_1}\cdots {\sf g}^{a_nb_n}\str\big(\ad(X_{b_1}\cdots \ad(X_{b_n})\big)
\end{equation}
and $\hC_{\ad} = \ad^{\otimes 2}(\hC)$.
As the adjoint representation of a simple Lie superalgebra is irreducible, then, by Schur's lemma, $\ad(C_k)$ is the scalar operator with the eigenvalue $c_k$, i.e. $\ad(C_k)=c_k I$ where $I$ is the identity operator acting in $V_{\ad}$.

Let us introduce the generating function for $c_k$:
\begin{equation}\label{cpzp}
c(z)=\sum_{p=0}^{\infty}c_pz^p.
\end{equation}
By \eqref{CpmMult} and \eqref{Ckad},
\begin{equation}\label{str2Cadzp1}
c(z)\cdot I=\str_2\big(\sum_{p=0}^\infty\hC_{\ad}^pz^p\big) =
\str_2\big(\sum_{p=0}^\infty\hCp^p z^p\big)+\str_2\big(\sum_{p=0}^\infty\hCm^p z^p\big),
\end{equation}
where we assume $\hC_{\pm}^0=\bP^{\ad}_{\pm}$, and $\hC_{\ad}^0=\bI$. By \eqref{Cm2byCm}, $\hCm^p=\left(-\frac{1}{2}\right)^{p-1}\hCm$, so
\begin{equation}\label{hCmzp}
\sum_{p=0}^\infty \hCm^pz^p=\bP^{\ad}_-+\sum_{p=1}^\infty \left(-\frac{1}{2}\right)^{p-1}\hCm z^p=\bPad_-+\frac{z}{1+\frac{z}{2}}\hCm.
\end{equation}
Now we express $\hCp^p$ in terms of $\proj^{(+)}(\alpha|\beta,\gamma)$, $\proj^{(+)}(\beta|\gamma,\alpha)$, $\proj^{(+)}(\gamma|\alpha,\beta)$ and $\proj^{(+)}_{(-1)}$,
that were defined in \eqref{vogproj} and \eqref{vogprojabg}. Using the condition
\begin{equation}
\label{polno}
\bPad_+ = \proj^{(+)}(\alpha|\beta,\gamma)+\proj^{(+)}(\beta|\gamma,\alpha)
+\proj^{(+)}(\gamma|\alpha,\beta)+\proj^{(+)}_{(-1)}
\end{equation}
yields
 $$
\begin{aligned}
 & \hCp^p = \hCp^p \Bigl(\proj^{(+)}(\alpha|\beta,\gamma)+
 \proj^{(+)}(\beta|\gamma,\alpha)
+\proj^{(+)}(\gamma|\alpha,\beta)+\proj^{(+)}_{(-1)}\Bigr) \\
&=\Bigl(-\frac{\alpha}{2t}\Bigr)^p\proj^{(+)}
(\alpha|\beta,\gamma)+\Bigl(-\frac{\beta}{2t}\Bigr)^p\proj^{(+)}
(\beta|\gamma,\alpha)+\Bigl(-\frac{\gamma}{2t}\Bigr)^p\proj^{(+)}
(\gamma|\alpha,\beta)+(-1)^p\proj^{(+)}_{(-1)} \; ,
\end{aligned}
$$
where $(-\frac{\alpha}{2t})$, $(-\frac{\beta}{2t})$,
$(-\frac{\gamma}{2t})$
and $(-1)$ are the eigenvalies of $\hCp$ corresponding to the projectors mentioned.
Therefore,
\begin{equation}\label{hCpzp1}
\begin{aligned}
\sum_{p=0}^\infty \hCp^pz^p&=\sum_{p=0}^\infty\bigg(\Big(-\frac{\alpha z}{2t}\Big)^p\proj^{(+)}(\alpha|\beta,\gamma)+\Big(-\frac{\beta z}{2t}\Big)^p\proj^{(+)}(\beta|\gamma,\alpha)\\
&+\Big(-\frac{\gamma z}{2t}\Big)^p\proj^{(+)}(\gamma|\alpha,\beta)+(-z)^p\proj^{(+)}_{(-1)}\bigg)\\
&=\frac{1}{1+\frac{\alpha z}{2t}}\proj^{(+)}(\alpha|\beta,\gamma)+\frac{1}{1+\frac{\beta z}{2t}}\proj^{(+)}(\beta|\gamma,\alpha)+\frac{1}{1+\frac{\gamma z}{2t}}\proj^{(+)}(\gamma|\alpha,\beta)+\\
&+\frac{1}{1+z}\proj^{(+)}_{(-1)}.
\end{aligned}
\end{equation}
By \eqref{polno}, \eqref{hCpzp1} can be rewritten as
\begin{equation}\label{hCpzp2}
\begin{aligned}
\sum_{p=0}^\infty\hCp^p z^p=&-\frac{\alpha z}{2t+\alpha z}\proj^{(+)}(\alpha|\beta,\gamma)-\frac{\beta z}{2t+\beta z}\proj^{(+)}(\beta|\gamma,\alpha)\\
&-\frac{\gamma z}{2t+\gamma z}\proj^{(+)}(\gamma|\alpha,\beta)-\frac{z}{1+z}\proj^{(+)}_{(-1)}+\bPad_+.
\end{aligned}
\end{equation}
Summing \eqref{hCmzp} and \eqref{hCpzp2} leads to
\begin{equation}\label{hCadzp1}
\begin{aligned}
\sum_{p=0}^\infty \hC_{\ad}^pz^p=&-\frac{\alpha z}{2t+\alpha z}\proj^{(+)}(\alpha|\beta,\gamma)-\frac{\beta z}{2t+\beta z}\proj^{(+)}(\beta|\gamma,\alpha)\\
&-\frac{\gamma z}{2t+\gamma z}\proj^{(+)}(\gamma|\alpha,\beta)-\frac{z}{1+z}\proj^{(+)}_{(-1)}+\frac{2z}{2+z}\hCm+\bI.
\end{aligned}
\end{equation}
To find $c(z)$ by \eqref{str2Cadzp1},
we need to calculate the supertrace $\str_2$ of the right hand side of
\eqref{hCadzp1}. Using \eqref{InvKilG},
 \eqref{CadGenC}, \eqref{IPKGenC} and \eqref{CpmDef}, we get the auxiliary supertraces
\begin{equation}\label{str2GenMany}
\begin{array}{c}
\str_2(\bI)=\sdim\mfg\cdot I, \;\;\;\;\; \str_2(\hCm)=-\frac{1}{2}I \; ,
 \;\;\;\;\; \str_2(\bPad)=I,  \\ [0.2cm]
 \str_2(\hCp) =\frac{1}{2}I, \;\;\;\;\; \str_2(\bK)=I,
\;\;\;\;\; \str_2(\hCp^2) =\frac{3}{4}I \; ,
\end{array}
\end{equation}
which are in accordance with (\ref{strGenMany}).
Then by \eqref{dimvog} and \eqref{str2GenMany},
\begin{equation}\label{str2Pabg}
\str_2\proj^{(+)}(\alpha|\beta,\gamma)=-\frac{(3\alpha-2t)t(\beta+t)(\gamma+t)}{\alpha(\alpha-\beta)(\alpha-\gamma)(\alpha-2t)}I.
\end{equation}
Using \eqref{dimvog}, \eqref{vogproj} and \eqref{str2GenMany} shows that
\begin{equation}\label{str2P-1}
\str_2\proj^{(+)}_{(-1)}=\frac{1}{\sdim \mfg}\str_2\bK=\frac{\alpha\beta\gamma}{(\alpha-2t)(\beta-2t)
(\gamma-2t)}I.
\end{equation}

Substituting \eqref{hCadzp1} into \eqref{str2Cadzp1} and using \eqref{str2GenMany}, \eqref{str2Pabg}, \eqref{str2P-1} results in
\begin{equation}
\begin{aligned}
\sum_{p=0}^\infty \hC_{\ad}z^p=&\sum_{p=0}^\infty c_pz^p\cdot I=\frac{(\alpha-2t)(\beta-2t)(\gamma-2t)}{\alpha\beta\gamma}\cdot I+\\
&z^2\frac{96t^3+168t^3z+6(14t^3+tt_2-t_3)z^2+
(13t+3tt_2-4t_3)z^3}{6(2t+\alpha z)(2t+\beta z)(2t+\gamma z)(2+z)(1+z)}\cdot I,
\end{aligned}
\end{equation}
where $t_2=\alpha^2+\beta^2+\gamma^2$ and $t_3=\alpha^3+\beta^3+\gamma^3$.
Therefore, the generating function for the eigenvalues of the higher Casimir operators of the $osp(M|N)$ and $s\ell(M|N)$ Lie superalgebras in the adjoint representation is
\begin{equation}
\label{resck}
\begin{aligned}
c(z)=&\frac{(\alpha-2t)(\beta-2t)(\gamma-2t)}{\alpha\beta\gamma}+\\
&z^2\frac{96t^3+168t^3z+6(14t^3+tt_2-t_3)z^2+
(13t+3tt_2-4t_3)z^3}{6(2t+\alpha z)(2t+\beta z)(2t+\gamma z)(2+z)(1+z)}.
\end{aligned}
\end{equation}
Formula \eqref{resck} is in agreement with the results of \cite{Mkr}, where this expression was found by using the formula for the values of the higher Casimir operators obtained in \cite{Okub}.

\section{Conclusion}
We have found explicit formulas for the projectors onto the invariant subspaces of the tensor product of two adjoint representations of the $osp(M|N)$ for $M-N\neq 0,1,2$ and $s\ell(M|N)$ Lie superalgebras for $M-N\neq 0,\pm 1,\pm 2$. The construction was performed by finding the characteristic identities for the split Casimir operator of the corresponding algebras. In the case of the $s\ell(M|N)$ Lie superalgebras, an additional ad-invariant operator was defined by means of the so-called symmetric structure constants of $s\ell(M|N)$. It was also shown that the dimensions of the invariant subspaces and the values of the quadratic Casimir operator in those subspaces are in agreement with \cite{Vog}--\cite{Lan2}, where these quantities are written by means of the Vogel parameters in the context of the universal Lie algebra. Furthermore, the generating function of the eigenvalues of the higher Casimir operators in the adjoint representation was found and expressed in terms of the Vogel parameters. The last result is in accordance with \cite{Mkr}.

\section*{Acknowledgements}
The authors are thankful to S.O.Krivonos and R.L.Mkrtchyan
  for useful discussions.
A.P.I. acknowledges the support of the Russian Science Foundation, grant No. 19-11-00131.

\end{document}